\documentclass[a4paper, 11pt]{article}
\pdfoutput=1
\usepackage{etoolbox, xstring,marvosym,geometry,fullpage}
\providebool{csub}
\setbool{csub}{false}
\usepackage{amsmath, amssymb, amsfonts}

\usepackage{nicefrac}
\usepackage{tgpagella}
\usepackage{mdwlist}

\usepackage{pdfsync}

\usepackage{amsthm}
\usepackage[colorlinks=true,citecolor=blue]{hyperref}
\usepackage[noadjust]{cite}
\usepackage{graphicx}
\usepackage{caption}
\usepackage{subcaption}

\newtheorem{theorem}{Theorem}[section]
\newtheorem{observation}[theorem]{Observation}
\newtheorem{claim}[theorem]{Claim}

\newtheorem{lemma}[theorem]{Lemma}

\theoremstyle{definition}
\newtheorem{definition}{Definition}[section]

\newtheorem*{remark}{Remark}

\newcommand{\abs}[1]{\ensuremath{\left|#1\right|}}

\newcommand{\defeq}[0]{\ensuremath{:=}}

\newcommand{\pdiff}[2][]{\ensuremath{\frac{\partial#1}{\partial#2}}}

\newcommand{\inb}[1]{\ensuremath{\left\{#1\right\}}}
\newcommand{\inp}[1]{\left(#1\right)}
\newcommand{\insq}[1]{\left[#1\right]}

\newcommand{\R}[0]{\ensuremath{\mathbb{R}}}
\newcommand{\etal}[0]{\emph{et al}}

\newcommand{\D}[1]{\ensuremath{\mathcal{D}_{#1}}}
\newcommand{\bip}[1]{\ensuremath{\mathrm{Bip}(#1)}}
\newcommand{\dis}[1]{\ensuremath{d\inp{#1}}}
\newcommand{\mon}[0]{{\sc Monotone 2-SAT}}
\newcommand{\avge}[1]{\langle #1\rangle}
\newcommand{\ignore}[1]{}

\begin{document}
\title{Lee-Yang theorems and the complexity of computing averages}
\newcommand{\sincgrant}{NSF grant CCF-1016896} 
\author{Alistair Sinclair\thanks{Computer Science Division, UC
    Berkeley. Email: sinclair@cs.berkeley.edu. Supported in part by
    \sincgrant.}\and Piyush Srivastava\thanks{Computer Science
    Division, UC Berkeley. Email:
    piyushsriva@gmail.com. Supported by the Berkeley Fellowship
    for Graduate Study and \sincgrant.}}

%   \author{Alistair Sinclair\inst{1} \and Piyush Srivastava\inst{2}}
% \institute{Soda Hall, University of California
%   Berkeley, Berkeley, CA 94720,
%   USA.\\ \email{sinclair@cs.berkeley.edu} \and
%   (\Letter) Soda Hall, University of California Berkeley, Berkeley,
%   CA 94720, USA.\\Tel: +1-510-316-1239. \email{piyushsriva@gmail.com}.
% }

% \\\email{sinclair@cs.berkeley.edu}. \and Piyush Srivatava
%   (\Letter)\at Soda Hall, University of California Berkeley, Berkeley,
%   CA 94720, USA.\\Tel:
%   +1-510-316-1239. \email{piyushsriva@gmail.com}.}

%
\date{}
\maketitle
\begin{abstract}
  We study the complexity of computing average quantities related to
  spin systems, such as the \emph{mean magnetization} and
  \emph{susceptibility} in the ferromagnetic Ising model, and the
  \emph{average dimer count} (or average size of a matching) in the
  monomer-dimer model.  By establishing connections between the
  complexity of computing these averages and the location of the
  complex zeros of the partition function, we show that these averages
  are \#P-hard to compute, and hence, under standard assumptions,
  computationally intractable.  In case of the Ising model, our approach
  requires us to prove an extension of the famous Lee-Yang Theorem
  from the 1950s.
\end{abstract}

\section{Introduction}
\label{sec:introduction}
\subsection{Background}
\label{sec:background}
Many natural computational problems in combinatorics, statistics and statistical
physics can be cast in the following framework.  We are given as input a graph
$G=(V,E)$ which implicitly defines a set $\Omega=\Omega(G)$ of combinatorial
structures, or {\it configurations\/} (such as matchings in~$G$, or $k$-colorings
of its vertices).  A weight function $w:\Omega\to\R^{+}$ assigns a positive weight
to every element $\sigma\in\Omega$, giving rise to a probability distribution
$\pi(\sigma)=w(\sigma)/Z$; here the normalizing factor
$Z:=\sum_{\sigma\in\Omega} w(\sigma)$ is called the {\it partition function}, and
is typically not known.  (While it is easy to compute the weight $w(\sigma)$ for
any given~$\sigma$, computing the sum~$Z$ is usually hard.  Note that the size
of~$\Omega$ is typically exponential in the size of~$G$.)  

The final ingredient is a non-negative function, or {\it observable}, 
$f:\Omega\to\R^{+}\cup\{0\}$, so that $f(\sigma)$ is also easily computable for
each $\sigma\in\Omega$.  Our goal is to compute the {\it average\/} of~$f$ with
respect to~$\pi$, i.e., 
\begin{equation*}\label{eq:avgedef}
   \avge{f} := \sum_{\sigma}\pi(\sigma)f(\sigma) = \frac{\sum_\sigma w(\sigma)f(\sigma)}{Z}.
\end{equation*}
As illustration we present two classical examples, both of which we will develop
extensively in our later results.
\par\medskip\noindent
{\it Example~1: The Ferromagnetic Ising Model.}  Here the configurations are 
assignments of spin values $\{+,-\}$ to the vertices of~$G$, i.e., $\Omega=\{+,-\}^V$.
The weight of a configuration~$\sigma$ is
\begin{equation}
  \label{eq:32}
   w_I(\sigma) \defeq \beta^{d(\sigma)} \lambda^{p(\sigma)},
 \end{equation}
where $d(\sigma)$ is the number of {\it disagreements\/} in~$\sigma$ (i.e., the number
of edges $\{u,v\}\in E$ with $\sigma(u)\ne\sigma(v)$), and $p(\sigma)$ is the number of 
vertices $v\in V$ with $\sigma(v)=+$.  The model has two parameters: the {\it edge
potential}~$\beta$, satisfying $0<\beta\le 1$, which governs the strength of the interaction
between neighboring spins; and the {\it vertex activity}~$\lambda>0$, which specifies
the tendency for spins to be~$+$.  The probability distribution $\pi(\sigma)=w_I(\sigma)/Z_I$
is the familiar {\it Gibbs distribution}, and $Z_I:=Z_I\inp{G,\beta,\lambda}$ is the associated
partition function.  Note that when $\beta < 1$, this distribution
favors agreement between neighboring spins.  Similarly, the
distribution favors `$+$' (respectively, `$-$') spins when $\lambda >
1$ (respectively, when $\lambda < 1$).

An important observable here is the {\it magnetization}~$p(\sigma)$,
which is just the
number of $+$-spins in~$\sigma$.  Its average, the {\it mean
  magnetization}, is a  
fundamental quantity in statistical physics: $$
   \avge{p} \defeq \frac{\sum_\sigma w_I(\sigma)p(\sigma)}{Z_I}.  $$
Other widely studied averages include the {\it mean energy}~$\avge{d}$ (the average
size of the cut between $+$-spins and $-$-spins) and the {\it
  susceptibility} $\chi \defeq \avge{p^2} -\avge{p}^2$
(the variance of the magnetization).\qed
\par\medskip\noindent
{\it Example~2: Matchings, or the Monomer-Dimer Model.}  The configurations~$\Omega$ 
are  all matchings (independent sets of edges) in~$G$.  The weight of a matching~$\sigma$
is
\begin{equation}
  \label{eq:33}
   w_M(\sigma) \defeq \lambda^{u(\sigma)} \prod_{e\in\sigma} \gamma_e,
 \end{equation}
where $u(\sigma)$ is the number of unmatched vertices (\emph{monomers}) in~$\sigma$.
The parameter $\lambda>0$ is the vertex weight (or \emph{monomer activity}), while for each
edge $e\in E$, $\gamma_e$ is an edge weight (or \emph{dimer activity}).  The Gibbs distribution
$\pi(\sigma) = w_M(\sigma)/Z_M$ is a natural weighted distribution on matchings, and the
partition function $Z_M:=Z_M(G,\{\gamma_e\}_{e\in E},\lambda)$ is the weighted matching polynomial
of~$G$.  

A natural observable here is $u(\sigma)$, the number of unmatched vertices (or monomers).
Note that $(|V|-\avge{u})/2$ is just the average size of a (weighted)
matching in~$G$ (or equivalently, the average number of dimers).\qed
\par\medskip\noindent
Observe that the Ising model partition function~$Z_I$ may be written as a polynomial in~$\lambda$
(actually a bivariate polynomial in~$\lambda$ and~$\beta$): $$
   Z_I = \sum_{k=0}^{|V|} \alpha_k\lambda^k,\qquad
               \hbox{\rm where ${\displaystyle\alpha_k = \sum_{\sigma:p(\sigma)=k}\beta^{d(\sigma)}}$.}$$
The mean magnetization then becomes
\begin{equation}\label{eq:magratio}
   \avge{p} = \frac{\sum_k k\alpha_k\lambda^k}{Z_I} = \frac{\mathcal{D}Z_I}{Z_I},
\end{equation}
where $\mathcal{D}$ denotes the differential operator $\lambda\pdiff{\lambda}$.  Similarly, the
mean energy and susceptibility $\chi$ can be written 
\begin{equation}\label{eq:ensuscratio}
   \avge{d} =  \beta\frac{\pdiff{\beta}Z_I}{Z_I};\qquad
   \chi \defeq \avge{p^2} - \avge{p}^2 =  \frac{\mathcal{D}^2Z_I}{Z_I} - \inp{\frac{\mathcal{D}Z_I}{Z_I}}^2.
\end{equation}
For matchings, the partition function~$Z_M$ is nothing other than the
matching polynomial $$
   Z_M = \sum_{k=0}^{|V|} \alpha_k\lambda^k, $$
where $\alpha_k = \sum_{\sigma:u(\sigma)=k} \prod_{e\in \sigma}\gamma_e$ is a weighted sum over
matchings with $k$ unmatched vertices.  The average number of monomers is then
\begin{equation}\label{eq:matchratio}
   \avge{u} = \frac{\mathcal{D}Z_M}{Z_M},
\end{equation}
where $\mathcal{D}$  again denotes the differential operator $\lambda\pdiff{\lambda}$.  

Equations~(\ref{eq:magratio})--(\ref{eq:matchratio}), which express averages as the ratio
of some derivative of the partition function to the partition function itself, are in fact no
accident; they are a consequence of the fact that the Gibbs distribution takes the
form $w(\sigma) = \exp(-H(\sigma))/Z$, where the \emph{Hamiltonian} $H(\sigma)$ is a sum of
natural observables.

The subject of this paper is the computational complexity of computing
natural averages such as (\ref{eq:magratio})--(\ref{eq:matchratio}).
While the complexity of computing partition functions has been widely studied in
the framework of Valiant's class {\#}P of counting
problems\footnote{\#P is the natural counting analog of the complexity
  class NP of decision problems.  See Appendix~\ref{sec:an-overv-compl} for details.}
 (see, e.g.,
\cite{bul06, bulatov05, cai_complexity_2012, cai_graph_2010, dyegre00,
  golgrojerthu09}),  
we are not aware of any corresponding results for the exact
computation of averages.  In the approximate setting, by contrast, it
is well known that (at least for the wide class of self-reducible
problems, which includes all the examples above) approximate
computation of the partition function is polynomial time equivalent to
sampling (approximately) from the Gibbs
distribution~$\pi$~\cite{jervalvaz86}; and sampling from~$\pi$ clearly
allows us to 
approximate averages to any desired accuracy (this is because
the observables one is concerned with in these situations admit
\emph{a priori} absolute bounds which are polynomial in the input
size).

What if we are interested in exact computation?  It is tempting to argue that computing
an average as in, say, (\ref{eq:magratio}) is at least as hard as computing the partition
function~$Z_I$, because (\ref{eq:magratio}) is a rational function and thus by evaluating
it at a small number of points we could recover the numerator and denominator polynomials
by rational interpolation.  Since the partition function is
{\#}P-hard\footnote{\#P-hard functions are as hard to compute as any
  function in \#P, and hence, widely believed to be intractable.  See
  Appendix~\ref{sec:an-overv-compl} for details.} in almost all cases of
interest (including $Z_I$ and $Z_M$ above at all but trivial values of the parameters), 
we would be done.

The problem with this argument is that, viewed as polynomials in the  
variable~$\lambda$, $Z_I$ and its derivative $\mathcal{D}Z_I$ may have
common factors (equivalently, viewed as polynomials in the complex variable
$\lambda$, they might have common zeros); and in this case we are
clearly not able to recover~$Z_I$ by rational interpolation.  Indeed it seems hard \emph{a priori}
to rule out the possibility that non-trivial interactions between $Z_I$ and its derivative
could conspire to make the average much easier to compute than~$Z_I$ itself.
Thus we are naturally led to the following question:
\par\smallskip\noindent
{\bf Question:} {\it Is it possible for the partition function~$Z$ and its derivative to have common 
zeros?}\footnote{Note that a common zero of $Z$ and $\mathcal{D}Z$ corresponds to a 
repeated zero of~$Z$, so this question is equivalent to the question of whether 
$Z$ may have repeated zeros.}
\par\smallskip\noindent
If the answer is no, then we will be able to conclude that computing
the average is as hard as computing~$Z$ itself, and thus {\#}P-hard in
all interesting cases.  

The main goal of this paper is to carry through this program using
resolutions of the above question in several interesting cases.  Before
proceeding, we mention a possible alternative approach to deal with
the issue of repeated zeros.  Since a generic polynomial does not have
repeated zeros, one could try to argue that any given graph $G$ can
be perturbed so that its partition function has distinct zeros,
and so that the magnetization of the perturbed graph is close
to the magnetization of the original graph.  One could then perform
the interpolation operations with respect to the perturbed graph, and
hope that if the perturbations are small enough, then the reduction
still goes through.  Indeed, this is our intuition for why
the magnetization (and other averages) should be hard to compute.

However, it is not clear how to convert this intuition into a formal
proof: in addition to a rather involved error analysis, this would
require showing that the partition function of a ``perturbed'' Ising model
behaves like a generic polynomial with respect to the structure of its
zeros, which seems no easier than answering the Question above.   Our
approach sidesteps this issue by tackling the question directly, and
in addition establishes a non-trivial property of the zeros of the
partition function that may be of independent interest. 

\subsection{Contributions}
\label{sec:contributions}
The question of common zeros actually turns out to be a deeper issue of wider interest
in statistical physics and complex analysis.  The study of the zeros of the partition function
dates back to the work of Lee and Yang in 1952: the famous {\it Lee-Yang
Circle Theorem\/}~\cite{leeyan52b}
proves the remarkable fact that the zeros of the ferromagnetic Ising partition function~$Z_I$
always lie on the unit circle in the complex plane.
This classical theorem, which has
since been re-proved many times in different ways~\cite{asano_lee-yang_1970,suzuki_zeros_1971,newman_zeros_1974},
was developed initially as an approach to studying phase transitions, but has
since spawned a more global theory connected with the Laguerre-P\'olya-Schur
theory of linear operators preserving stability of polynomials.  (See the Related Work section
below.)

Somewhat surprisingly, despite much activity in this area, the question of 
the location of the zeros of the derivative $\mathcal{D}Z_I$ (or equivalently, of repeated zeros of~$Z_I$ itself) remains
open.  The main technical contribution of this paper is to resolve this question as follows.
\begin{theorem}
  \label{thm:lee-yang-ext}
  Let $G = (V,E)$ be a connected graph, and suppose $0 < \beta <
  1$.  Then the zeros of the polynomial $\mathcal{D}Z_I(G,\beta,\lambda)$ (in $\lambda$)
  satisfy $\abs{\lambda} < 1$.
\end{theorem}
Since the Lee-Yang Theorem says that all zeros of~$Z_I$ satisfy $|\lambda|=1$, 
Theorem~\ref{thm:lee-yang-ext} immediately implies that $Z_I$ and $\mathcal{D}Z_I$ have no
common zeros.  
\begin{remark}The restriction that $G$ be connected is needed:
  there exist disconnected graphs for which
  the conclusion of the theorem does not hold.  A simple example is a
  graph consisting of two isomorphic disconnected subgraphs.  For the
  same reason we require 
  $\beta<1$.  We also note that standard facts from complex analysis
  (in particular, the Gauss-Lucas theorem) imply that the zeros of~$\mathcal{D}Z_I$ lie in the
  convex hull of those of~$Z_I$, and hence within the closed unit
  circle.  The content of Theorem~\ref{thm:lee-yang-ext} is that they
  must lie in the {\it interior\/} of the circle.  This refinement is of course
  crucial for our application.
\end{remark}

Before moving on, let us briefly mention our approach to the proof.  We actually 
prove a more general result concerning the zeros of the multivariate partition
function $Z_I(G,\beta,\allowbreak\{\lambda_v\}_{v\in V})$; see Theorem~\ref{thm:wsglp} 
in Section~\ref{sec:an-extended-lee}.
(The Lee-Yang Theorem itself is also often stated in multivariate form.)
Our proof is based on a delightful combinatorial proof of the Lee-Yang Theorem
due to Asano~\cite{asano_lee-yang_1970}, which begins with the empty graph (which 
trivially satisfies the theorem) and builds the desired graph~$G$ by repeatedly adding
edges one at a time; by a careful induction one can show that the Lee-Yang
property is preserved under each edge addition.  Our proof follows a similar induction,
but the argument is more delicate because we are working with the
more complicated polynomial~$\mathcal{D}Z_I$ rather than~$Z_I$.  In particular, in the inductive
step we need to evoke a non-trivial correlation inequality due to 
Newman~\cite{newman_zeros_1974}.

Our first computational complexity result follows as an almost
immediate corollary
of Theorem~\ref{thm:lee-yang-ext}.
\begin{theorem}\label{thm:ising-main}
  For any fixed $0<\beta<1$
  and any fixed $\lambda\neq 1$, the problem of computing the mean
  magnetization of the Ising model on connected graphs is {\#}P-hard.
  Moreover, the problem remains {\#}P-hard even when the input is
  restricted to graphs of maximum degree at most $\Delta$, for any fixed
  $\Delta \geq 4$.  
\end{theorem}
Note that in the case $\lambda=1$ the mean magnetization is trivially
$|V|/2$ by symmetry.  Theorem~\ref{thm:ising-main} confirms that in
all non-trivial cases, the problem of computing the fundamental
average quantity associated with the Ising model is as hard as it
could possibly be. Furthermore, the result also holds for bounded
degree graphs, which are relevant in the statistical physics setting.
The result can also be extended to arbitrary ferromagnetic two-spin
systems and to planar graphs: the
details can be found in
\notbool{csub}{Appendix~\ref{sec:hardness-general-two}}{the full 
  version}.

We also prove a similar (but slightly weaker) result for the
\emph{susceptibility} of the Ising model.
\begin{theorem}\label{thm:ising-main-suscpet}
  For any fixed $0<\beta<1$,
  the problem of computing the susceptibility of the Ising model on
  connected graphs, when $\lambda$ is specified in unary, is
  {\#}P-hard.  Moreover, the problem remains \#P-hard even when the
  input is restricted to graphs of maximum degree at most $\Delta$ for
  any $\Delta \geq 3$. 
\end{theorem}
\begin{remark}
  The requirement that $\lambda$ be part of the input seems to be an
  artifact of the rational interpolation operations we use in our
  proof.  In particular, our proof of Theorem~\ref{thm:ising-main}
  shows hardness for \emph{fixed} $\lambda$ by ``simulating''
  different values of $\lambda$ by suitably modifying the graph.   To
  adapt this reduction approach to prove hardness for susceptibility
  (at fixed values of $\lambda$) seems to require the polynomial time
  computation of 
  magnetization as a subroutine.  However, we conjecture that
  computing the susceptibility should be hard even for fixed values of
  $\lambda$ (including $\lambda = 1$).  
\end{remark}

We then proceed beyond the Ising model, and ask about the hardness of
computing averages in the monomer-dimer model (i.e., weighted
matchings).  A classical result of Heilmann and
Lieb~\cite{heilmann_theory_1972} establishes an analog of the Lee-Yang
Theorem for the zeros of the monomer-dimer partition function~$Z_M$;
however, Heilmann and Lieb also present examples of (connected)
graphs~$G$ for which~$Z_M$ has repeated zeros, so we cannot hope to
prove an analog of Theorem~\ref{thm:lee-yang-ext} in this case.  On
the other hand, Heilmann and Lieb show that if $G$ contains a
Hamiltonian path then all the zeros of~$Z_M$ are simple.  We are able
to capitalize on this fact by adapting existing 
{\#}P-hardness reductions for $Z_M$ in such a way that the instances of~$Z_M$ that
appear in the reduction always contain a Hamiltonian path.
Specifically, we present a reduction from the problem \mon\ of
counting satisfying assignments of a monotone 2-CNF formula to
computing~$Z_M$ in Hamiltonian graphs~$G$.  The reduction is an
elaboration of Valiant's original {\#}P-completeness proof for the
permanent~\cite{val79a}.

This leads to our third computational complexity result.
\begin{theorem}\label{thm:monomer-dimer-main}
  For any fixed $\lambda>0$, the problem of computing the average
  number of dimers (equivalently, average size of a matching) in the
  monomer-dimer model on connected graphs with edge
  weights in the set $\{1,2,3\}$  is {\#}P-hard.  Moreover,
  the problem remains \#P-hard even when the input is restricted to
  graphs of maximum degree at most $\Delta$, for any $\Delta \geq 5$.
\end{theorem}
\begin{remark}
  Note that our hardness result requires a small finite number (three)
  of different values for the edge weights.  However, this requirement
  can be removed if $G$ is allowed to have parallel edges; the theorem
  then holds for any single fixed non-zero edge weight (including the
  uniform case in which all edge weights are $1$).
\end{remark}
\subsection{Related Work}
\label{sec:related}
The study of the location of zeros of the partition function was
initiated by Yang and Lee~\cite{leeyan52} in connection with the
analysis of phase transitions.  In the follow-up
paper~\cite{leeyan52b}, they instantiated this approach for the
ferromagnetic Ising model by proving the celebrated Lee-Yang theorem
on the location of zeros of the partition function and using it to
conclude that the ferromagnetic Ising model can have at most one phase
transition.  The Lee-Yang approach has since become a
cornerstone of the study of phase transitions, and has been used
extensively in the statistical physics literature: see, e.g.,
\cite{asano_lee-yang_1970,heilmann_theory_1972,
  newman_zeros_1974,suzuki_zeros_1971} for specific examples, and Ruelle's
book~\cite{ruelle83:_statis_mechan} for background.  In a
slightly different line of work, Biskup
\etal~\cite{bisborCKK04,biskup_partition_2004B} 
studied a novel approach to Lee-Yang theorems for a general class of
spin systems on lattice graphs using asymptotic expansions of the
partition function.  Zeros of partition functions have also been
studied in a purely combinatorial setting without reference to the
physical interpretation: see, for example, Choe
\etal~\cite{choe_homogeneous_2004} for a collection of such results
about zeros of a general class of partition functions.  Another
important example is the work of Chudnovsky and
Seymour~\cite{chudnovsky_roots_2007}, who show that the zeros
of the independence polynomial of claw-free graphs lie on the real
line.  There have also been attempts to relate the Lee-Yang program to
the Riemann hypothesis~\cite{newman_ghs_1991}.  

Lee-Yang theorems have also been studied in mathematics in connection
with the theory of \emph{stability preserving} operators.  The main
problem underlying this area is the characterization of linear
operators that preserve the class of polynomials, called
\emph{$\Omega$-stable} polynomials, whose zeros lie in some fixed
closed set $\Omega$.  This research area has its origins in the work of
Laguerre~\cite{laguerre82:_sur} and of P\'{o}lya and
Schur~\cite{polya14:_uber_arten_faktor_theor_gleic}, and also has
connections to control theory~\cite{clark_routh-hurwitz_1992} and to
electrical circuit theory~\cite{brune_synthesis_1931}. 
It has also seen considerable recent activity, especially through
the breakthrough results of Borcea and Br\"{a}nd\'{e}n, who completely
characterize stability preserving operators for multivariate
polynomials in various important
settings~\cite{borcea_lee-yang_2009,borcea_polya-schur_2009}.
Although the study of stability preserving operators is closely related to our
problem, there is a crucial difference in that we require our linear
differential operator to not only preserve the stability of the
partition function, but in fact to \emph{improve} it, by restricting
the possible locus of the zeros of the derivative to the \emph{open}
interior of the locus of the zeros of the partition function itself.  

In the statistical physics literature, we are aware of
only two works which consider the multiplicity of the zeros of the
Ising partition function: Heilmann and
Lieb~\cite{heilmann_theory_1972} and Biskup
\etal~\cite{bisborCKK04,biskup_partition_2004B}.  In
~\cite{heilmann_theory_1972}, a theorem similar to our
Theorem~\ref{thm:lee-yang-ext} is proven in the special case when the
underlying graph $G$ has a Hamiltonian path and $\beta$ is close
enough to $1$ (depending upon the graph $G$).  Similarly, in
the special case of the Ising model, the results of \cite{bisborCKK04}
imply our result but only when $\beta$ is close to $0$, and only in
the special case of lattice
graphs~\cite{biskup_partition_2004B,biskup12}.  Note that neither of
these results appears to be sufficient for the purposes of our hardness
result.

The classification of counting problems associated with partition
functions (via so-called \emph{dichotomy theorems}) has also recently
been a very active area of research.  For several interesting general
classes of partition functions, these theorems characterize the
partition function as being either computable in polynomial time or 
\#P-hard~\cite{cai_complexity_2012,golgrojerthu09,cai_graph_2010,bul06,bulatov05,dyegre00}.
However, there appear to be no analogous results on the
complexity of averages such as the magnetization.

A related area that we do not deal with in this work is the problem of
approximate counting.  Recent progress in this area has shown that the
complexity of approximating the partition function, as well as that of
the related problem of approximate sampling, is closely related to the
phase transition
phenomenon~\cite{Weitz06CountUptoThreshold,Sly2010CompTransition,sly12}.
However, it is not clear whether hardness results analogous to
\cite{Sly2010CompTransition,sly12} can be proven for the 
approximate computation of the magnetization.

\section{Preliminaries}
\label{sec:preliminaries}
\subsection{The models}
\label{sec:models}
Let $G = (V,E)$ be an undirected graph.  The two models we will be
concerned with are the \emph{ferromagnetic Ising model} and the
\emph{monomer dimer model}, both of which have already been defined
in Section~\ref{sec:background}.

\noindent\textbf{Ferromagnetic Ising model.} Recall that in the
\emph{ferromagnetic Ising model}, a \emph{configuration} $\sigma: V
\rightarrow \inb{+,-}$ is an assignment of $+/-$ spins to the vertices
of $G$.  The model is characterized by an \emph{edge potential}
$0 < \beta \leq 1$, and a \emph{vertex activity} $\lambda
> 0$.  The weight function $w_I(\sigma)$ defined in (\ref{eq:32})
induces a probability distribution over configurations with an
associated partition function $Z_I(G, \beta, \lambda) := \sum_\sigma
w_I(\sigma)$.  We shall be concerned with the \emph{mean
  magnetization} $M(G,\beta, \lambda) \defeq \avge{p}$, which is the
average number of $+$-spins in a configuration, and the
\emph{susceptibility} $\chi \defeq \avge{p^2}-\avge{p}^2$, which is
the variance of the same quantity.  As in
(\ref{eq:magratio})-(\ref{eq:ensuscratio}), these quantities can be
written in terms of the derivatives of $Z_I$ with respect to
$\lambda$. %

For our discussion of the zeros of $Z_I(G,\beta,\lambda)$, we
will also need a generalization of the Ising model in which the
vertex activities can vary across vertices of G.  Suppose that the
vertex activity at vertex $v$ is $z_v$.  The weight of a
configuration $\sigma$ is then defined as
\begin{align*}
  w_I(\sigma) &\defeq \beta^{\dis{\sigma}}\prod_{v:\sigma(v) = +}z_v,
\end{align*}
and the partition function is given by
$Z_I(G,\beta,\inp{z_v}_{v\in V}) = \sum_{\sigma}w_I(\sigma)$.
Consider the linear differential operator ${\cal D}_G$ defined as
follows:
\begin{equation*}
  \D{G} \defeq \sum_{v \in V} z_v\frac{\partial}{\partial z_v}. 
\end{equation*}
As in (\ref{eq:magratio})--(\ref{eq:ensuscratio}), we can then write
the magnetization $M_I(G,\allowbreak\beta,\inp{z_v}_{v\in V})$ as 
\begin{equation*}
  M(G,\beta,\inp{z_v}_{v\in V}) = \frac{{\cal D}_GZ_I(G,\beta,\inp{z_v}_{v\in V})}{Z_I(G,\beta,\inp{z_v}_{v\in V})}.
\end{equation*}

\noindent\textbf{Monomer-dimer model.} Recall that in the
\emph{monomer-dimer model}, the configurations are matchings of $G$.
The model is characterized by \emph{edge weights} $\gamma_e > 0$ for every edge
$e$ in $E$ and a \emph{vertex activity} $\lambda > 0$.  The weight
$w_M(\sigma)$ of a matching $\sigma$ is as described in (\ref{eq:33}),
and the associated partition function is defined by $Z_M(G,
\inp{\gamma_e}_{e\in E},\lambda) \defeq \sum_\sigma w_M(\sigma)$.  

The average number of monomers $U(G,\inp{\gamma_e}_{e\in
  E},\allowbreak\lambda)  \defeq \avge{u}$ can be written (as in
(\ref{eq:matchratio})) in terms of the 
derivative of $Z$: in particular $U(G,\inp{\gamma_e}_{e\in E},
\lambda) = \frac{\mathcal{D}Z_M(G,\inp{\gamma_e}_{e\in E},
  \lambda)}{Z_M(G,\inp{\gamma_e}_{e\in E},\lambda)}$.  The
\emph{average dimer count} $D(G,\inp{\gamma_e}_{e\in E},\lambda)$
(equivalently, the average size of a matching) can be obtained from
$U$ by the simple relation 
\begin{displaymath}
  D(G,\inp{\gamma_e}_{e\in E},\lambda) = \frac{n - U(G,\inp{\gamma_e}_{e\in E},\lambda)}{2},
\end{displaymath}
where $n$ is the number of vertices in $G$. 

\begin{remark}
  In our definitions above, vertex activities are restricted to be
  positive real numbers.  Although this is the physically (and
  computationally) relevant setting, in our proofs and in our discussion
  of Lee-Yang theorems we will need to work with vertex activities
  that are arbitrary complex numbers.  The expressions for the
  quantities defined above still remain valid.
\end{remark}

\subsection{Zeros of partition functions}
\label{sec:lee-yang-heilmann}
We first consider the location of the complex zeros of the partition
function of the ferromagnetic Ising model.  In a seminal paper Lee and
Yang proved the following striking theorem~\cite{leeyan52b}.
\begin{theorem}[\cite{leeyan52b}]
  \label{thm:lee-yang-1} Let $G$ be any undirected graph and suppose
  $0 < \beta \leq 1$.  Then the complex zeros of $Z_M(G,\beta,
  z)$, considered as a polynomial in $z$, satisfy
  $\abs{z} = 1$.  
\end{theorem}
Actually, Lee and Yang proved the following multivariate version of
their theorem, the proof of which was later considerably simplified by
Asano~\cite{asano_lee-yang_1970}.
\begin{theorem}[\cite{leeyan52b,asano_lee-yang_1970}]
  \label{thm:lee-yang-2} Let $G=(V,E)$ be a connected undirected
  graph, and suppose $0 < \beta < 1$.  Suppose $(z_v)_{v \in V}$ is
  a set of complex valued vertex activities such that $\abs{z_v} \geq
  1$ for all $v\in V$, and $\abs{z_u} > 1$ for at least one
  $u \in V$.  Then $Z_I(G, \beta, \inp{z_v}_{v\in V}) \neq 0$.
\end{theorem}
Theorem~\ref{thm:lee-yang-2} is readily seen to imply
Theorem~\ref{thm:lee-yang-1} by setting $z_v = z$ for all $v \in V$.
We now consider the partition function of the monomer-dimer model.  In
\ifbool{csub}{\cite{heilmann_theory_1972}, Heilmann and Lieb proved
  the following result}{\cite{heilmann_monomers_1970}, Heilmann and
  Lieb stated  the following result (see \cite{heilmann_theory_1972}
  for the complete proof)}.
\begin{theorem}[\ifbool{csub}{\cite{heilmann_theory_1972}}{\cite{heilmann_monomers_1970,heilmann_theory_1972}}]
  \label{thm:heilmann-lieb} Let $G=(V,E)$ be any graph, and
  $(\gamma_e)_{e\in E}$ be a collection of positive real edge weights. The
  complex zeros of $Z_M(G, (\gamma_e)_{e\in E}, z)$, considered as a
  polynomial in $z$, satisfy $\Re\inp{z} = 0$.  Further, if $G$ contains a
  Hamiltonian path, all the zeros are simple.
\end{theorem}
In \cite{heilmann_theory_1972}, Heilmann and Lieb also gave examples
of \emph{connected} graphs $G$ in which $Z_M$ has repeated zeros.
This is in contrast to the Ising model, where, as we prove in
Section~\ref{sec:an-extended-lee}, connectedness is sufficient to
ensure that the zeros are simple.

\subsection{Rational interpolation}
\label{sec:rati-interp}

In our hardness reductions, we will need a few well known
facts about interpolation of rational functions.  While it is clear
that it is not in general possible to determine all coefficients of a
rational function given its values at any number of points, this can be
done if we impose a few simple conditions, as stated in the following
theorem.
\begin{theorem}[\cite{macon_existence_1962}]
  \label{thm:rational-interp}
  Suppose $R(x) = \frac{p(x)}{q(x)}$ where $\mathrm{gcd}(p(x),
  \allowbreak q(x)) = 1$
  and both $p(x)$ and $q(x)$ are of degree $n$.  Suppose
  $\tilde{p}(x)$ and $\tilde{q}(x)$ are polynomials of degree at most
  $n$ satisfying
    \begin{displaymath}
      \frac{\tilde{p}(x_i)}{\tilde{q}(x_i)} =  R(x_i)
    \end{displaymath}
    for $2n+2$ distinct values $x_1,x_2,\ldots x_{2n+2}$.  Then there
    is a constant $c$ such that $p(x) = c\tilde{p}(x)$ and $q(x) =
    c\tilde{q}(x)$. 
\end{theorem}
Notice that given the evaluations at the points $x_i$ one
can write down a system of $2n+2$ homogeneous linear equations for the
$2n+2$ unknown coefficients of $p$ and $q$.  The theorem then
guarantees that this system has rank exactly $2n+1$.  Thus, since
Gaussian elimination can be implemented to run in strongly polynomial
time (see, e.g., \cite{edmonds67:_system_distin_repres_linear_algeb}),
a polynomial time algorithm for evaluating $R$ immediately yields a
polynomial time algorithm for determining \emph{some} $\tilde{p}$ and
$\tilde{q}$ satisfying the conditions of the above theorem.  If we
know at least one non-zero coefficient of $p$ or $q$, we can then
determine the proportionality constant $c$, and hence $p$ and $q$
also, in time polynomial in $n$.

\section{An extended Lee-Yang theorem}
\label{sec:an-extended-lee}

In this section we prove Theorem~\ref{thm:lee-yang-ext}, our extension
of the classical Lee-Yang theorem.
Let $G = (V,E)$ be a connected graph with $|V| = n$ and $|E|= m$, with
vertex activity $z_i$ at the $i$th vertex.  When clear from the
context , we will write $Z(G)$ and $M(G)$ for the partition function
$Z_I(G, \beta, (z_v)_{v\in V})$ and the mean magnetization $M(G,
\beta, (z_v)_{v\in V})$ of the Ising model on $G$.   In terms of the
linear operator $\D{G}$ defined in Section~\ref{sec:models}, we
then have $M(G) = \D{G}Z(G)/Z(G)$.  

For convenience, we will use the shorthand $Y' = \D{G}Y$ (when $G$ is
clear from the context) in this section.  Notice that this is slightly
non-standard, as this shorthand is usually used for the actual
derivative.  In particular, when all the $z_i$ are equal to $z$, we
have $Y' = z\pdiff[Y]{z}$ with our notation.  Also, observe that the
operator $\D{G}$ obeys the usual product rule: $(Y_1Y_2)' = Y_1'Y_2
+ Y_1Y_2'$.

In our proof, we will also need the following generalization of the
partition function.  We call an assignment of positive integer valued
\emph{weights} $w:V \rightarrow \mathbb{Z}^+$ to the vertices of $G$ \emph{legal}
if $w(v)$ is at least equal to the degree of $v$, for all $v \in V$.  

\begin{definition}
  Let $w$ be a legal collection of weights for $G$.  The
  \emph{weighted partition function} $Z_w(G)$ is then defined as
\begin{equation}
Z_w(G) \defeq \sum_{\sigma \in \inb{+,-}^V}\beta^{d(\sigma)}\prod_{v: \sigma(v) = +}z_v^{w(v)},\label{eq:1}
\end{equation}
where, as before, $d(\sigma)$ is number of disagreeing edges in the
configuration $\sigma$. 
\end{definition}
\noindent Notice that the multivariate Lee-Yang theorem
(Theorem~\ref{thm:lee-yang-2}) holds also for the weighted partition
function, since all the weights are positive
integers and we are effectively just changing variables from $z_v$ to
$z_v^{w(v)}$.  

We will also need the following consequence of a correlation
inequality of Newman~\cite{newman_zeros_1974}, whose proof can be found
in \notbool{csub}{Appendix~\ref{sec:proof-crefthm:newman}.
}{the full version.
}
\begin{theorem}[{\cite[Theorem 3.2]{newman_zeros_1974}}]\label{thm:newman}
  Let $G$ be any graph, and let $w$ be a legal collection of weights
  for $G$.  Suppose $0 < \beta < 1$, and $\abs{z_v} \geq 1$ for all $v \in
  V$ are such that $Z_w(G) \neq 0$.  Then\footnote{Recall that we are
    using here the slightly non-standard notation $Z_w'(G) =
    \D{G}Z_w(G)$, as described at the beginning of this section.}
  \begin{displaymath}
    \Re\inp{M(G)} = \Re\inp{\frac{Z_w'(G)}{Z_w(G)}} \geq n/2. 
  \end{displaymath}
  Here, $\Re\inp{z}$ denotes the real part of $z$.
\end{theorem}
\noindent In the special case of real valued activities, the above theorem is
equivalent to the well known Griffiths inequality
\cite{griffiths_correlations_1967}, which states the intuitive fact
that in a ferromagnetic Ising model where all activities favor the $+$ spin, the
magnetization must be at least $n/2$. 

For ease of reference in the inductive proof, we give a name to the
property we want to establish.  Recall that when all the
vertex activities are equal to $z$, the classical Gauss-Lucas theorem,
together with the Lee-Yang theorem, implies that the zeros of the
derivative $\D{G}(Z(G))$ lie \emph{on or inside} the unit
circle.  Our goal is to establish that they actually lie \emph{inside}
the unit circle.  Accordingly, we use the following terminology:
\begin{definition}[Strict Gauss-Lucas property]
  A graph $G=(V,E)$ has the \emph{strict Gauss-Lucas property (SGLP)}
  if for every set of activities such that $\abs{z_v} \geq 1$ for all
  $v\in V$, and every $0 < \beta < 1$, one has $\D{G}Z(G) \neq 0$.  The
  graph has the \emph{weighted strict Gauss-Lucas property (WSGLP)} if
  for all legal weights $w$, $\D{G}Z_w(G) \neq 0$ necessarily holds
  under the same conditions.
\end{definition}
Note that WSGLP easily implies SGLP: we simply choose $w(v) = \Delta$
for all $v$,
where $\Delta$ is the maximum degree of $G$.  From
WSGLP, we then have that whenever $\abs{z_v}\geq 1$ and $0< \beta <
1$, $\D{G}Z_w(G) = \Delta\D{G}Z(G) \neq 0$, and hence $\D{G}Z(G) \neq
0$.  Thus Theorem\nobreakspace \ref {thm:lee-yang-ext} is implied by the following more
general statement.
\begin{theorem}\label{thm:wsglp}
  Every connected graph has the weighted strict Gauss-Lucas property,
  and hence also the strict Gauss-Lucas property.
\end{theorem}
We now proceed to prove Theorem\nobreakspace \ref {thm:wsglp}, using
induction on the 
number of edges in the graph $G$.  We first consider the base case of
a connected graph with a single edge.
\begin{lemma}[Base Case]
  \label{lem:base}
  Let $G$ be the graph consisting of a single edge connecting two
  vertices.  Then $G$ has the weighted strict Gauss-Lucas property. 
\end{lemma}
\begin{proof}
  In this case we have $Z_w(G) = z_1^{w_1}z_2^{w_2} + \beta(z_1^{w_1}
  + z_2^{w_2}) + 1$ and therefore $\D{G}Z_w(G) = (w_1 + w_2)z_1^{w_1}z_2^{w_2} +
  \beta(w_1z_1^{w_1} + w_2z_2^{w_2})$, with $w_1,w_2 \geq 1$.  When
  $\abs{z_1},\abs{z_2}\geq 1$, the latter vanishes only if
  \begin{displaymath}
   w_1 + w_2 = \beta\abs{\frac{w_2}{z_1^{w_1}} + \frac{w_1}{z_2^{w_2}}} \leq
   \beta\inp{\frac{w_2}{\abs{z_1^{w_1}}} + \frac{w_1}{\abs{z_2^{w_2}}}} \leq  \beta\inp{w_1 + w_2},
  \end{displaymath}
  which cannot hold since $0 < \beta < 1$.  
\end{proof}
For the inductive case, we require two operations: adding a new vertex
to the graph, and merging two existing vertices.  These operations are
formalized in the following lemmas.

\begin{lemma}[Adding a vertex]
  \label{lem:add-vertex}
  Suppose $G = (V, E)$ is a connected graph satisfying the weighted
  Gauss-Lucas property.  Let $u$ be a vertex not in
  $V$.  Then, the graph $G_1$ obtained by attaching the new vertex $u$
  to any vertex (say $v_1$) of $G$ also has the weighted
  Gauss-Lucas property.
\end{lemma}

\begin{lemma}[Merging vertices]
  \label{lem:add-edge}
  Suppose $G=(V,E)$ is a connected graph satisfying the weighted
  strict Gauss-Lucas property.  Consider any two vertices, say $v_1$
  and $v_2$, in $G$ that are not connected by an edge.  The graph $G_1$
  obtained by merging $v_1$ and $v_2$ into a single
  vertex $v$ (while making all the edges incident on $v_1$ and $v_2$
  incident on $v$) also has the weighted strict Gauss-Lucas property.
\end{lemma}
Before proceeding with the proofs of the above lemmas, we show
how to use them to prove Theorem\nobreakspace \ref {thm:wsglp}.
\begin{proof}[Proof of Theorem\nobreakspace \ref {thm:wsglp}]
  We will prove by induction on $m$ that any connected graph with at
  most $m$ edges satisfies WSGLP.  By Lemma\nobreakspace \ref
  {lem:base}, this statement is true when $m = 1$. Now suppose that
  the statement is true when $m = k$, and consider any connected graph
  $G$ with $k + 1$ edges.

  In case $G$ has a cycle, there exist
  vertices $u$ and $v$ such that the edge $\inb{u,v}$ can be removed
  from $G$ to obtain a connected graph $H$.  Since $H$ has at most $k$
  edges, $H$ satisfies WSGLP by the inductive hypothesis.  Let $v_1$
  be a vertex not in $G$.  By Lemma\nobreakspace \ref {lem:add-vertex}, the graph
  $H\cup\inb{\inb{u,v_1}}$ satisfies WSGLP.  We can now merge $v_1$
  and $v$ to obtain $G$, which therefore satisfies WSGLP by
  Lemma\nobreakspace \ref {lem:add-edge}.

  In case $G$ is a tree, there exists an edge
  $\inb{u,v}$ such that $v$ is of degree $1$.  Again, we obtain a
  connected graph $H$ with at most $k$ edges by removing the edge
  $\inb{u,v}$.  By the inductive hypothesis, $H$ satisfies WSGLP, and
  hence by Lemma\nobreakspace \ref {lem:add-vertex}, $G$ does too.  This completes the
  induction.
\end{proof}
\begin{remark}
  Note that the proof of Theorem~\ref{thm:wsglp} given above holds
  also when the graph $G$ is allowed to have parallel edges and
  self-loops.  This will be useful in our extension to general
  two-state ferromagnetic spin systems\notbool{csub}{ in
    Appendix~\ref{sec:hardness-general-two}.   
  }{, which can be found in the full version.}
\end{remark}

We turn now to the proofs of Lemmas\nobreakspace \ref {lem:add-vertex}
and\nobreakspace  \ref {lem:add-edge}, for which we will need the
following additional lemma.
\begin{lemma}
  \label{lem:rec} Let $G$ be a connected graph.  Fix any set $S$ of
  vertices of $G$, and let $Z_w^+(S)$ denote the partition function
  restricted to configurations on the subgraph $G-S$, with all the vertices in $S$
  fixed to have spin $+$.  Consider any set of vertex activities
  satisfying $\abs{z_v} \geq 1$ for $v\in G - S$.  Then, for $0 < \beta < 1$ and any set of
  permissible weights on the vertices of $G$, we have $Z_w^+(S) \neq 0$
  and $\Re\inp{Z_w^+(S)'/Z_w^+(S)} \geq 0$.  In particular, for any
  positive real $a$, we have $Z_w^+(S)' + aZ_w^+(S) \neq 0$.
\end{lemma}
\begin{proof}
  Observe that $Z_w^+(S)$
  is proportional to the product of weighted partition functions on
  connected components of the graph $G-S$, where the activities on the
  vertices connected to $S$ in these components (of which there is at
  least one in each component) have \emph{increased} in magnitude by a
  factor of at least $1/\beta > 1$.  We can therefore conclude using
  Theorem\nobreakspace \ref {thm:lee-yang-2} that $Z_w^+(S) \neq 0$.
  The second condition $\Re\inp{Z_w^+(S)'/Z_w^+(S)} \geq 0$ then
  follows from Theorem\nobreakspace \ref {thm:newman} applied to $G-S$.
\end{proof}
We first prove Lemma\nobreakspace \ref {lem:add-edge}, since its proof
is somewhat simpler.
\begin{proof}[Proof of Lemma\nobreakspace \ref {lem:add-edge}]
  Consider any legal weight assignment on $G_1$.  If the weight of $v$
  in $G_1$ is $w_v$, we can write $w_v = w_1 + w_2$ such that the
  weight assignment giving weights $w_1$ and $w_2$ to $v_1$ and $v_2$
  respectively is legal for $G$.  By partitioning into four cases
  based on the spins of $v_1$ and $v_2$, we can write the
  corresponding weighted partition function $Z_w(G)$ and its
  derivative as
  \begin{align}
    Z_w(G) &= Az_1^{w_1}z_2^{w_2} + C{z_1}^{w_1} + D{z_2}^{w_2} +
    B;\label{eq:3}\\
    \notbool{csub}{Z_w(G)' &= (A' + (w_1 + w_2)A)z_1^{w_1}z_2^{w_2} + (C' +
    w_1C)z_1^{w_1} + (D' + w_2D)z_2^{w_2} + B',\label{eq:4}}{
    Z_w(G)' &= (A' + (w_1 + w_2)A)z_1^{w_1}z_2^{w_2} + (C' +
    w_1C)z_1^{w_1}\nonumber\\
    &\quad+ (D' + w_2D)z_2^{w_2} + B',\label{eq:4}
    }
  \end{align}
  for polynomials $A$, $B$, $C$, $D$ in the remaining variables $z_i$.
  Notice that in the notation of Lemma~\ref{lem:rec}, $A =
  Z_w^+(\inb{v_1,v_2})$.  Similarly, denoting the activity at the
  merged vertex by $z$, we have the following expressions for $G_1$:
  \begin{align}
    Z_w(G_1) &= Az^{w_1+w_2} + B;\label{eq:5}\\
    Z_w(G_1)' &= (A' + (w_1 + w_2)A)z^{w_1+w_2} + B',\label{eq:6}
  \end{align}
  with $A$ and $B$ as defined above.  Now consider any fixing of the
  activities such that $|z_i| \geq 1$ for $i > 2$.  Since $G$
  satisfies the weighted strict Gauss-Lucas property, we get by
  setting $z_1 = z_2$ in eq.\nobreakspace \textup {(\ref {eq:4})} that
  the (univariate) polynomial
  \begin{displaymath}
    (A' + (w_1 + w_2)A)z^{w_1 + w_2} + (C' + w_1C)z^{w_1} + (D' + w_2D)z^{w_2} + B'
  \end{displaymath}
  in $z$ has no zeros satisfying $\abs{z} \geq 1$.  Also, we know from
  Lemma\nobreakspace \ref {lem:rec} that $A' + (w_1 + w_2)A \neq 0$.  Thus, we must have
  that the product of the zeros, ${B'}/{(A' + (w_1 + w_2)A)}$, satisfies
  \begin{displaymath}
    \abs{\frac{B'}{A' + (w_1 + w_2)A}} < 1.  
  \end{displaymath}
  However, using eq.\nobreakspace \textup {(\ref {eq:6})}, this implies that if $\abs{z_i} \geq 1$
  for $i > 2$, then $Z_w(G_1)'$ can be zero only if $\abs{z} < 1$, and
  hence $G_1$ satisfies the weighted strict Gauss-Lucas
  property.
\end{proof}
Finally, we give the proof of Lemma\nobreakspace \ref
{lem:add-vertex}.
\begin{proof}[Proof of Lemma\nobreakspace \ref {lem:add-vertex}]
  Note that any legal set of weights for $G_1$ can be obtained by
  adding one to the weight $w_1$ of $v_1$ in a legal set of weights
  $w$ of $G$, and then assigning $u$ an arbitrary weight $w_0 \geq 1$.
  With a slight abuse of notation, we denote these related weight
  assignments (one on $G$ and the other on $G_1$) by the same letter
  $w$.  We now partition the terms in $Z_w(G)$ based on the spin of
  $v_1$ to get
  \begin{align*}
    Z_w(G) &= Az_1^{w_1} + B;\\
    Z_w(G)' &= (A' + w_1A)z_1^{w_1} + B',
  \end{align*}
  where $w_1$ is the weight of $v_1$ in $G$.  Here, $A$, $B$ are
  polynomials in the remaining variables $z_i$, and $A$ is of the form
  $Z_w^+(\inb{v_1})$ in the notation of Lemma~\ref{lem:rec}.  We again
  assume $0 < \beta < 1$ and $\abs{z_i} \geq 1$ for $i > 1$.  We
  now consider $G_1$.  Denoting the activity at $u$ by $z$, we can
  write
  \begin{align*}
      Z_w(G_1) &= A(\beta + z^{w_0})z_1^{w_1+1} + B(1 + \beta
      z^{w_0});\\
      \ifbool{csub}{
        Z_w(G_1)' &= (A' + w_1A)(\beta + z^{w_0})z_1^{w_1+1} 
        +  {w_0}\beta B z^{w_0}\\
        &\quad+ A(\beta + ({w_0}+1)z^{w_0})z_1^{w_1+1} 
        + B'(1 + \beta z^{w_0}).
      }{
        Z_w(G_1)' &= (A' + w_1A)(\beta + z^{w_0})z_1^{w_1+1} + A(\beta
        + ({w_0}+1)z^{w_0})z_1^{w_1+1} + B'(1 + \beta z^{w_0})\\
        &\qquad+
        {w_0}\beta B z^{w_0}.
      }
  \end{align*}
  Now suppose that $G_1$ does not satisfy the weighted strict
  Gauss-Lucas property, and hence $\abs{z}$ and $\abs{z_1}$ are both
  also at least $1$, but $Z_w(G_1)' = 0$.  It follows from
  Theorem\nobreakspace \ref {thm:newman} that we then also have
  $Z_w(G_1) = 0$.  We now proceed to derive a contradiction to the
  above observations.  For convenience, we denote $z_1^{w_1+1}$ by $y$
  in what follows.

  Using Lemma\nobreakspace \ref{lem:rec}, we know that $A \neq 0$ and that $A' + w_1A
  \neq 0$ for our setting of activities.  By Theorem\nobreakspace \ref
  {thm:lee-yang-2} applied to $Z_w(G)$ and the
  weighted strict Gauss-Lucas property applied to $Z_w(G)'$, we get
  \begin{equation}
    \abs{\frac{B}{A}} \leq 1,\quad\text{ and }\quad\abs{\frac{B'}{A' + w_1A}} < 1.\label{eq:8}
  \end{equation}
  Also, since $Z_w(G_1) = 0$, we must have
  \begin{align}
    y &= -\frac{B}{A}\frac{1 + \beta z^{w_0}}{\beta + z^{w_0}}.\label{eq:12}
  \end{align}
  Notice that $y$ is well defined since $A\neq 0$, $\abs{z} \geq 1$ and
  $\beta < 1$.  Further, since $\beta < 1$, either one of $\abs{z} > 1$,
  or $\abs{B} < \abs{A}$ would imply that $\abs{y} < 1$, which is a
  contradiction to our assumption that $\abs{z_1}\geq 1$ (since $y =
  z_1^{w_1 + 1}$).  Thus, we must
  have
  \begin{equation}
    \abs{z} = 1,\quad\text{ and }\quad\abs{\frac{B}{A}} = 1.\label{eq:13}
  \end{equation}
  Now, substituting the value of $y$ from eq.\nobreakspace \textup {(\ref {eq:12})} into $Z_w(G_1)' =
  0$, we get 
  \begin{displaymath}
    \ifbool{csub}{
      \begin{split}
        &B'(1 + \beta z^{w_0}) + \beta {w_0}Bz^{w_0} =\\
        &\quad\inp{\inp{A' +
            w_1A}(\beta + z^{w_0}) +A(\beta +
        ({w_0}+1)z^{w_0})}\frac{B}{A}\frac{1+\beta z^{w_0}}{\beta +
        z^{w_0}}.
    \end{split}
  }{
    B'(1 + \beta z^{w_0}) + \beta {w_0}Bz^{w_0} = \inp{\inp{A' +
        w_1A}(\beta + z^{w_0}) +A(\beta +
      ({w_0}+1)z^{w_0})}\frac{B}{A}\frac{1+\beta z^{w_0}}{\beta +
      z^{w_0}}.
  }
  \end{displaymath}
  Dividing through by $(A' + w_1A)(1 + \beta z^{w_0})$, setting $c =
  A/(A'+w_1A)$ and rearranging terms, we get 
  \begin{align}
    \frac{B'}{A' + w_1A} &= \frac{B}{A}\inp{1 + c + {w_0}c\inp{\frac{z^{w_0}}{\beta+z^{w_0}} +
        \frac{1}{1 + \beta z^{w_0}} - 1}}\nonumber\\
    \notbool{csub}{&= \frac{B}{A}\inb{1 + c + {w_0}c\inp{2\Re\inp{\frac{z^{w_0}}{\beta+z^{w_0}}}- 1}}\text{, since
      $\abs{z} =1$}.\label{eq:2}}{
    &= \frac{B}{A}\inb{1 + c +
      {w_0}c\inp{2\Re\inp{\frac{z^{w_0}}{\beta+z^{w_0}}}-1}}\label{eq:2}
  }
  \end{align}
  \ifbool{csub}{since $\abs{z} =1$.}{} 
  Notice that these divisions are  well defined since $A' + w_1A \neq 
  0$, and $\beta < 1$ and $\abs{z} = 1$ implies that $(1 + \beta
  z^{w_0}) \neq 0$ as well.  Note also that $c$ is of the form
    $1/(w_1 + c')$ where $\Re(c') = \Re\inp{A'/A} \geq 0$ by
    Lemma~\ref{lem:rec} and our earlier observations
    about $A$: it therefore follows that $\Re(c) \geq 0$.
  However, we then calculate that for $|z| =1$, the factor inside the
  braces in (\ref{eq:2}) has real part (and hence absolute value) at
  least $1$.  Using $\abs{B}/\abs{A} = 1$ from (\ref {eq:13}), we then
  see that the right hand side of (\ref {eq:2}) always has
  absolute value at least 1, which gives us the required contradiction
  to (\ref {eq:8}).  This shows that $G_1$ satisfies the weighted
  strict Gauss-Lucas property.
\end{proof}

\notbool{csub}{
\section{Hardness of computing the mean magnetization}}{
\section{Hardness of computing the\\mean magnetization}
}
\label{sec:hardn-comp-magn}
In this section, we use our extended Lee-Yang theorem
(Theorem~\ref{thm:lee-yang-ext}) to  prove
\notbool{csub}{Theorems\nobreakspace \ref 
{thm:ising-main} and \ref{thm:ising-main-suscpet} via reductions}{
 Theorem~\ref{thm:ising-main} via a reduction 
} from the problem of computing the partition function of the Ising
model, which is known to be \#P-hard even for bounded degree
graphs~\cite{dyegre00,bulatov05}.  More specifically, we will use the
following \#P-hardness result.
\begin{theorem}[{\cite[Theorem 1]{bulatov05}},{\cite[Theorem 5.1]{dyegre00}}]
\label{thm:ising-partition-hard}
Fix $\beta$ satisfying $0 < \beta < 1$.  The
problem of computing the partition function $Z_I(G,\beta,1)$ of
the Ising model on connected graphs of fixed maximum degree $\Delta
\geq 3$ is \#P-hard.
\end{theorem}
For simplicity, we prove here a version of
Theorem~\ref{thm:ising-main} without the bounded degree constraint.
The extension to bounded degree graphs requires some more work and is proved in
\notbool{csub}{Appendix~\ref{sec:bdd-degree-hard}.
}{the full version.
}
\begin{proof}[Proof of Theorem\nobreakspace \ref
  {thm:ising-main}]
  We assume $\lambda > 1$, since the case $\lambda < 1$ is
  symmetrical.  For given $0 < \beta < 1$, suppose that we have an
  algorithm $\mathcal{A}$ which,
  given a connected graph $G$, outputs the mean magnetization
  $M(G,\beta,\lambda)$ in polynomial time.  Let $G$ be a graph of $n$
  vertices.  Notice that as a rational function in $z$, $M(G, \beta,
  z)$ is a ratio of the two polynomials, ${\cal D}Z_I(G,\beta,z)$ and
  $Z_I(G,\beta,z)$, which are both of degree $n$.  Further, since $G$
  is connected, these polynomials are co-prime by Theorem\nobreakspace
  \ref {thm:lee-yang-ext}.  Thus, if we could efficiently evaluate
  $M(G, \beta, z)$ at $2n + 2$ distinct points $z$ using algorithm
  ${\cal A}$, we could uniquely determine the coefficients of
  $Z_I(G,\beta,z)$ by Theorem\nobreakspace \ref {thm:rational-interp}
  (since we know that the constant term in $Z_I(G,\beta,z)$ is $1$).
  We could then determine $Z_I(G,\beta, 1)$ in polynomial
  time. Theorem\nobreakspace \ref {thm:ising-partition-hard} would
  then imply that computing the mean magnetization
  for the given values of the parameters $\beta$ and $\lambda$ is \#P-hard.
  
  In order to evaluate $M(G, \beta, z)$ at $2n + 2$ distinct values,
  we consider the graph $G(k)$ obtained by attaching $k$ new neighbors
  to each vertex of $v$.  We then have
  \begin{align}
    Z_I(G(k), \beta, \lambda) &= (1+\beta\lambda)^{nk}Z_I(G, \beta,
    \lambda_k)\text{, and}\label{eq:25}\\
    M(G(k), \beta, \lambda) &= \frac{kn\beta\lambda}{1 +\beta\lambda}
    + \insq{ 1 + \frac{k\lambda(1-\beta^2)}{(1+\beta\lambda)(\beta+\lambda)}}M(G,\beta,\lambda_k),\label{eq:24}
  \end{align}
  where $\lambda_k =
  \lambda\inp{\frac{\beta+\lambda}{1+\beta\lambda}}^k$.  Notice that
  when $\beta < 1$, all the $\lambda_k$ are distinct, and further,
  $M(G,\beta,\lambda_k)$ can be easily determined given
  $M(G(k),\beta,\lambda)$.  Therefore, we can evaluate $M(G(k), \beta,
  \lambda)$ for $0 \leq k \leq 2n + 1$ using the algorithm
  $\mathcal{A}$, and then using eqs.\nobreakspace \textup {(\ref
    {eq:25})} and\nobreakspace \textup {(\ref {eq:24})} we can
  determine $M(G,\beta, \lambda_k)$ in polynomial time.  Since these
  evaluations are at distinct points, the reduction is complete.
\end{proof}
\notbool{csub}{\begin{proof}[Proof of Theorem\nobreakspace \ref
  {thm:ising-main-suscpet}]
  For a given $\beta$ as specified in the theorem, suppose that there
  is a polynomial time algorithm ${\cal A}$ which, given a graph $G$
  of maximum degree $\Delta$, and a value of $\lambda$ in unary,
  outputs the susceptibility $\chi(G,\beta,\lambda)$.  Notice that as
  a rational function in $z$, $\chi(G, \beta, z)$ is a ratio of the
  two polynomials $Z_I(G,\beta,z)\cdot{\cal D}^2Z_I(G,\beta,z) -
  ({\cal D}Z_I(G,\beta,\lambda))^2$ and $Z_I(G,\beta,z)^2$, which are
  both of degree $2n$.  Further, since $G$ is connected, these
  polynomials are co-prime by Theorem\nobreakspace \ref
  {thm:lee-yang-ext}.  To see this, notice that any common complex
  zero of these two polynomials must be a common zero of
  $Z_I(G,\beta,\lambda)$ and ${\cal D}Z_I(G,\beta,\lambda)$, which is
  prohibited by Theorem~\ref{thm:lee-yang-ext}.

  To complete the reduction, we notice that we can choose $4n+2$
  distinct values of $\lambda$ in the interval $(0,1]$ all of which
  have a unary representation length of at most $5n$.  Thus, using
  ${\cal A}$, we can efficiently evaluate $\chi(G, \beta, z)$ at $4n +
  2$ distinct values of $z$.  By Theorem\nobreakspace \ref
  {thm:rational-interp} we can then use these evaluations to uniquely
  determine the coefficients of $Z_I(G,\beta,z)^2$ (since we already
  know that the constant coefficient is $1$), and hence,
  $Z_I(G,\beta, 1)$, in polynomial time.  Because of
  Theorem\nobreakspace \ref {thm:ising-partition-hard}, this implies
  that the problem of computing the susceptibility at the given value
  of $\beta$ is \#P-hard.
  \end{proof}
}{The proof of the hardness of susceptibility
  (Theorem~\ref{thm:ising-main-suscpet}) is similar in flavor and can
  be found in the full version.}%

\ifbool{csub}{
\section{Hardness of computing the\\ average dimer count}}{
\section{Hardness of computing the average dimer count}
}
\label{sec:hardn-comp-aver}
In this section we prove Theorem\nobreakspace \ref
{thm:monomer-dimer-main} by reducing the \#P-hard problem \#{\sc
  Monotone-2SAT} to the problem of computing the average dimer count.
The reduction is similar in structure to Valiant's original proof for
the \#P-hardness of the problem of counting perfect matchings.
However, since we will need to do rational interpolation, we need the
zeros of the partition function to be simple, so by
Theorem~\ref{thm:heilmann-lieb} we will need to ensure that the graph
appearing as the output of the reduction always has a Hamiltonian
path.  The formal properties satisfied by our reduction are stated in
the following theorem.
\begin{theorem}
  \label{thm:monomer-reduction-main}
  There exists a polynomial time algorithm $\mathcal{A}$ which,
  when given as input a {\sc Monotone 2-SAT} formula $\phi$,
  outputs a weighted graph $G$ with the
  following properties:
  \begin{enumerate*}
  \item The weights in $G$ are drawn from the set $\inb{1,2,3}$. \label{item:1}
  \item Suppose $\phi$ has $\nu$ variables and $\mu$ clauses. Then,
    given the total weight $W$ of perfect matchings in $G$, the number
    of satisfying assignments of $\phi$ can be determined in
    polynomial time from $W$, $\mu$, and $\nu$.  \label{item:2}
  \item $G$ contains a Hamiltonian path.  \label{item:3}
  \end{enumerate*}
\end{theorem}
We observe here that Valiant's reduction from \#3-SAT~\cite{val79a}
can be easily
modified so that it satisfies properties \ref{item:1} and
\ref{item:2}.  However, it is property \ref{item:3} that is crucial
for our purposes, since it allows the use of Theorem\nobreakspace \ref
{thm:heilmann-lieb}.  We first show how Theorem
\ref{thm:monomer-reduction-main} can be used to immediately prove a
slightly weaker version of Theorem~\ref {thm:monomer-dimer-main},
which shows hardness only on general graphs. The proof showing
hardness for bounded degree graphs can be found in
\notbool{csub}{Appendix~\ref{sec:bdd-degree-hard}.
}{the full version.
}
  \begin{proof}[Proof of Theorem\nobreakspace \ref
    {thm:monomer-dimer-main}]
    Fix any $\lambda > 0$, and suppose that there exists a polynomial
    time algorithm $\mathcal{B}$ which, given a connected graph $H$,
    with edge weights in the set $\inb{1,2,3}$ outputs
    $D(H,(\gamma_e)_{e\in E},\lambda)$.  In the following, we suppress
    the dependence on edge weights $\inp{\gamma_e}_{e\in E}$ for
    clarity of notation.  Given a {\sc Monotone 2-SAT} formula $\phi$,
    we can then produce the graph $G = \mathcal{A}(\phi)$ in
    polynomial time.  Let $n$ be the number of vertices in $G$.  Since
    $G$ contains a Hamiltonian path, Theorem\nobreakspace \ref
    {thm:heilmann-lieb} implies that $Z_M(G,z)$ and ${\cal D}Z_M(G,z)$
    have no common zeros.  Thus, being able to use algorithm ${\cal
      B}$ to evaluate $D(G,z)$ (and hence $U(G, z)$) at $2n + 2$
    different values of $z$ would allow us to uniquely determine the
    coefficients of $Z_M(G,z)$ in polynomial time by rational
    interpolation (Theorem~\ref{thm:rational-interp}), since we
    already know that the coefficient of $z^n$ is $1$.  This would
    allow us to obtain $W$ (which is the constant term in $Z_M(G,z)$),
    and hence, by property~\ref{item:2}, also the number
    of satisfying assignments of $\phi$, in polynomial time.  This
    would show that the problem of computing ${\cal D}(G,\lambda)$ is
    \#P-hard~(since \#{\sc Monotone-2SAT} is \#P-hard~\cite{val79b}).

    However, $\mathcal{B}$ only allows us to evaluate $U(G,z)$ at $z
    = \lambda$.  In order to ``simulate'' other values of $\lambda$,
    we consider the graph $G(k)$ obtained by attaching $k$ new vertices
    to each vertex of $G$ with unit weight edges.  We then have
    \begin{align}
      Z_M(G(k), \lambda) &= \lambda^{nk}Z_M(G,\lambda_k);\label{eq:22}\\
      U(G(k), \lambda) &= nk + \frac{\lambda^2 - k}{\lambda^2 + k}U(G,\lambda_k),\label{eq:23}
    \end{align}
    where $\lambda_k = \lambda + k/\lambda$.  Thus, by choosing $2n+2$
    different values of $k$, none of which is equal to $\lambda^2$, we
    can determine $U(G,z)$ at $2n + 2$ different values of $z$ by
    running $\mathcal{B}$ on $G(k)$ and using eq.\nobreakspace
    \textup {(\ref {eq:23})}.  This completes the proof.
  \end{proof}

  In the rest of this section, we proceed to
  \notbool{csub}{prove}{sketch the steps in the proof of}
  Theorem\nobreakspace \ref
  {thm:monomer-reduction-main} \notbool{csub}{in a sequence 
  of steps}{(the full proof can be found in the full version)}.  For
simplicity, we will describe our reduction in terms 
  of cycle covers in a directed graph rather than perfect matchings in
  an undirected graph (this also allows us to directly compare our
  gadget construction with that of Valiant~\cite{val79a} at various
  steps).  Given a weighted directed graph $G = (V, E)$, we define the
  undirected bipartite graph $\bip{G} = (V \times \inb{0,1}, E')$
  where the edge $\inb{(x, 0), (y, 1)}$ is in $E'$ with weight
  $\gamma_e$ if and only if $(x,y)$ is an edge in $E$ with the same
  weight.  Note that a subset $S \subseteq E$ forms a cycle cover of
  weight $w$ in $G$ if and only if the corresponding subset of edges
  $S' = \inb{\inb{(x,0), (y,1)}| (x,y) \in S}$ forms a perfect
  matching of weight $w$ in \bip{G}.  In particular, the total weight
  of all perfect matchings in \bip{G}~is the same as the total weight
  of all cycle covers of $G$.
  
\notbool{csub}{
  Later, while arguing about the existence of Hamiltonian paths in
  graphs of the form \bip{G}, we will find it convenient to use the
  following short-hand notation for simple paths in the graph
  \bip{G}~in terms of the edges and vertices of $G$.  Consider any
  simple path $(x_1, 1)$, $(x_2, 0)$, $(x_3, 1)$, $(x_4, 0), \ldots
  (x_l, 1)$, where we have assumed for simplicity that $l$ is odd.
  The edges corresponding to this path in $G$ are $x_1 \leftarrow
  x_2$, $x_2\rightarrow x_3$, $x_3\leftarrow x_4, \ldots x_{l-1}
  \rightarrow x_l$.  Notice that alternate edges are traversed in
  reverse in this representation.  The path can therefore be
  represented as $x_1 \leftarrow x_2 \rightarrow x_3 \leftarrow x_4
  \rightarrow \ldots \rightarrow x_l$.  Similarly for a path starting
  on the other side, say $(x_1, 0), (x_2, 1), (x_3, 0), (x_4, 1)$, we
  have the representation $x_1\rightarrow x_2 \leftarrow x_3
  \rightarrow x_4$.  Notice that a path $p_2$ starting at a vertex $v$
  in this notation can be appended to a path $p_1$ ending at $v$ if
  and only if the arrows at $v$ in $p_1$ and
  $p_2$ respectively are in opposite directions.   We will refer to
  this notation as the \emph{alternating path} representation.
  Further, given an alternating path representation of a path, we will
  refer to edges going right (such as $x_1\rightarrow x_2$ in the
  last example) as \emph{forward} edges, and edges going left (such
  as $x_2\leftarrow x_3$ in the above example) as \emph{backward}
  edges.
}{
 
  We also note that paths in \bip{G}~correspond to \emph{alternating
    paths} in $G$: paths in which alternate edges are traversed in
  reverse.  For example, the path $(x_1, 1)$, $(x_2, 0), \ldots, (x_l,
  1)$ in $G$, where we assume for simplicity that $l$ is odd,
  corresponds to the alternating path $x_1 \leftarrow x_2$,
  $x_2\rightarrow x_3$, $x_3\leftarrow x_4, \ldots x_{l-1} \rightarrow
  x_l$ in $G$.  This correspondence will allow us to argue about the
  existence of a Hamiltonian path in $\bip{G}$ in terms of a (suitably
  defined) \emph{alternating} Hamiltonian path in $G$.  See the full
  version for the formal definition.  

}

  \subsection{Overview of the reduction}
  \label{sec:overv-our-reduct}
  We now look at the basic structure of our reduction, which is an
  elaboration of Valiant's reduction~\cite{val79a} as modified by
  Papadimitriou~\cite{pap94} and presented in
  \cite{arora09:_comput_compl}.  Recall 
  that given a {\sc Monotone 2-SAT} formula $\phi$, the reduction
  needs to produce in polynomial time a directed graph $G$ such that
  the number of satisfying assignments of $\phi$ can be easily
  determined from the total weight of cycle covers of $G$, and such
  that $\bip{G}$ has a Hamiltonian path.   Our first step is to
  introduce a shared variable in all the clauses of $\phi$: this
  shared variable will be useful later in showing the existence of a
  Hamiltonian path through the gadget.  
  \begin{observation}
    Let $\phi = \bigwedge_{i=1}^\mu c_i$ be a \mon~formula with $\mu$
    clauses, $\nu$ variables, and $s$ satisfying
    assignments.  Let $\tau$ be a variable not appearing in
    $\phi$ and consider the 3-SAT formula
    \begin{displaymath}
      \phi' = \mathop{\bigwedge}_{i=1}^\mu(\tau \lor c_i).
    \end{displaymath}
    The number of satisfying assignments of $\phi'$ is $s'\defeq 2^\nu + s$. 
  \end{observation}
  Notice that each clause in $\phi'$ has exactly three variables, and
  that the number of satisfying assignments of $\phi$ can be easily
  determined given the number of satisfying assignments of $\phi'$.
    
  We start the construction of $G$ by creating a separate
  \emph{variable gadget} (see 
  Figure~\ref{fig:var-gadget-1}) for each of the variables $\tau$, $x_1$,
  $x_2$, $\ldots$, $x_\nu$ occurring in $\phi'$.  This gadget has an external
  \emph{dotted} edge for each appearance of the variable in the
  formula, and is designed so that any cycle cover must
  either use \emph{all} the dotted edges in a particular gadget, or
  none of them.   
  \begin{figure}[h]
    \centering
    \ifbool{csub}{\includegraphics[scale=0.5]{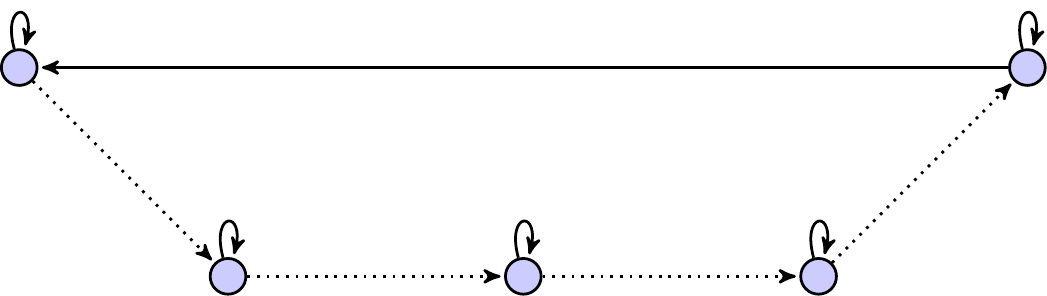}}{
    \includegraphics[scale=0.7]{tikz-var-gadget-1}}

    \caption{Variable gadget}
    \label{fig:var-gadget-1}
  \end{figure}

  As done in Valiant's reduction, we then introduce a \emph{clause gadget} (see
  Figure~\ref{fig:3-clause-gadget-m}) for each clause in $\phi'$.
  Each clause gadget has one external dotted edge for each literal in
  the clause, and is designed so that no cycle cover can include all the
  dotted edges; and so that for any other subset of the dotted edges,
  there is exactly one cycle cover including all the edges in the subset
  and no others.  For each clause gadget, we label each of
  the three dotted edges in the gadget with one of the three literals
  appearing in the clause.  However, in this step, we ensure
  that in each gadget the $b\rightarrow c$ dotted edge is the one
  labeled with the literal $\tau$, since this is needed to
  show that the final construction has a Hamiltonian path.  We now
  ``pair'' each dotted edge appearing in a clause gadget with a dotted
  edge corresponding to the same literal in a variable gadget, so that
  each dotted edge appears in exactly one pair.  
  
  We first consider cycle covers which obey the constraint that they
  must choose \emph{exactly} one edge from each such pair.  We claim
  that the number of cycle covers satisfying this ``pairing''
  constraint equals the number of satisfying assignments of $\phi'$.
  To see this, we associate a truth assignment with every cycle cover
  by setting the variable $v$ to \emph{true} if the cycle cover uses
  all the dotted edges in the variable gadget for $v$, and to
  \emph{false} if it uses none of the dotted edges.  Notice that
  because of the pairing constraint, a cycle cover is uniquely
  determined by specifying its assignment.  Further, given the above
  properties of the clause gadget, exactly those cycle covers are
  permitted whose associated assignments are satisfying assignments of
  $\phi'$.
  \begin{figure}[h]\centering
    \ifbool{csub}{
      \includegraphics[scale=0.4]{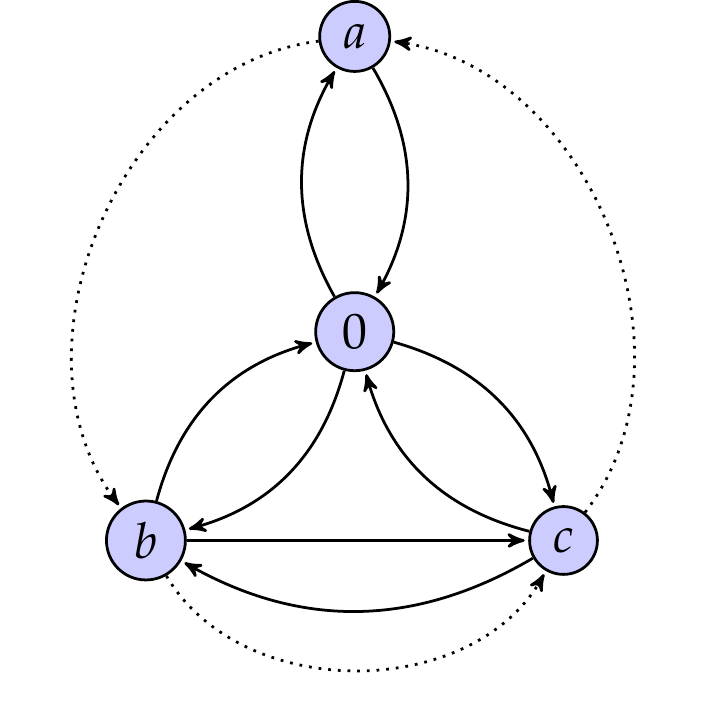}
    }{
      \includegraphics[scale=0.6]{tikz-3-clause-gadget}
    }
    \caption{3-SAT clause gadget}
    \label{fig:3-clause-gadget-m}
  \end{figure}
  
  We now enforce the ``pairing'' constraint referred to above using a
  gadget similar to Valiant's \emph{XOR-gadget}.  The XOR-gadget has
  two \emph{ports} (labeled $a$ and $d$), each of which admits one
  incoming and one outgoing edge (see
  Figure~\ref{fig:xor-gadget-main}).  To ensure the
  ``pairing'' constraint for a pair of dotted edges $e_1\rightarrow
  f_2$ and $e_2\rightarrow f_2$, we replace them by the
  incoming-outgoing pair of a single XOR-gadget (see
  Figure~\ref{fig:use-xor-gadget}).  The gadget has the property that
  after the replacement, the weight of every cycle cover which would
  have included exactly one of the two dotted edges $e_1\rightarrow
  f_1$ and $e_2\rightarrow f_2$ in the original graph gets multiplied
  by a factor of $2$ (for each replacement made), while the weight of
  any cycle covers not satisfying the pairing constraint becomes $0$
  (see \notbool{csub}{Appendix~\ref{sec:xor-gadget}}{the
    full version} for a proof). The total weight of all cycle covers
  in the final graph so obtained is therefore $2^{l}s'$, where $s'$ is
  the number of satisfying assignments and $l$ is the total
  number of literals in $\phi'$ (since one XOR-gadget is needed to
  replace the pair of dotted edges for each literal).  Further,
  replacing a pair of edges by a XOR gadget does not change the
  in-degree or out-degree of any vertex already present.

  Note that the XOR-gadget has edges of weight $-1$, which
  are not permitted in the monomer-dimer model.  This can be remedied
  by replacing the $-1$ weight edges by a large chain of edges (of
  length, say, $m^2$ where $m$ is the number of edges in the original
  graph) of weight 2, with individual vertices in the chain having
  self loops (of weight 1).  The total weight of cycle covers in the
  new graph modulo $2^{m^2} + 1$ then gives the total weight of cycle
  covers in the original graph.

  This last step of replacing the $-1$ edge by a long chain presents a
  challenge since we will need to include all the vertices in the
  chain in our Hamiltonian path (equivalently, all $-1$ weight edges
  must appear in the Hamiltonian path).  For this reason, we cannot
  use Valiant's XOR-gadget directly.  Our XOR-gadget, on the other
  hand, is such that the $-1$ weight edges can always be included in
  our Hamiltonian path.  However, we have to be careful in the
  orientation of the XOR-gadgets \notbool{csub}{in order to be able to
    construct a Hamiltonian path later}{for this to hold}: when
  replacing a pair of dotted edges one of which belongs to $\tau$'s
  variable gadget, we orient the XOR-gadget so that the incoming edge
  at vertex $a$ in the XOR-gadget comes from the variable gadget.  At
  all other pairs, 
  we orient the XOR-gadgets so that the incoming edge at the vertex
  $a$ comes from a clause gadget.
  \begin{figure}[ht]
    \centering
      \begin{subfigure}[b]{0.4\textwidth}
        \centering
        \ifbool{csub}{
          \includegraphics[scale=0.5]{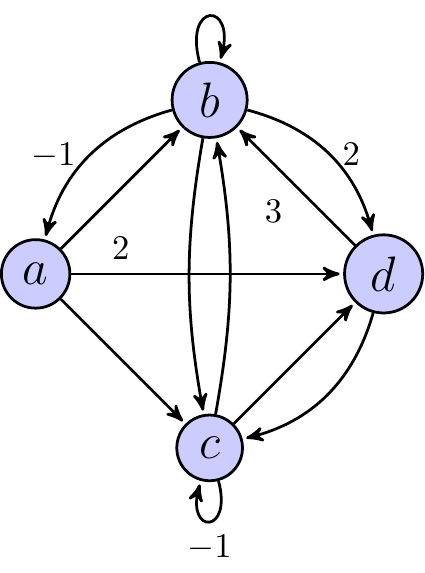}
        }{
          \includegraphics[scale=0.8]{tikz-xor-gadget}
        }
        \caption{XOR-gadget}\label{fig:xor-gadget}
      \end{subfigure}
      \begin{subfigure}[b]{0.4\textwidth}
        \centering
        \ifbool{csub}{
          \includegraphics[scale=0.37]{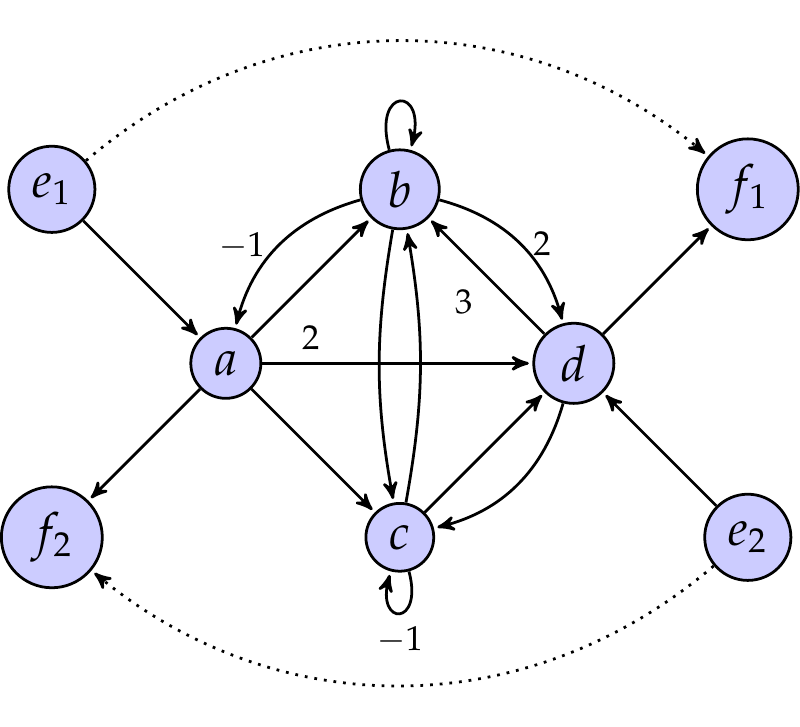}
        }{
          \includegraphics[scale=0.6]{tikz-xor-gadget-use}
        }
        \caption{Use of XOR-gadget}
        \label{fig:use-xor-gadget}
      \end{subfigure}
    \caption{Replacing dotted edges $e_1\rightarrow
          f_1$ and $e_2\rightarrow f_2$}
        \label{fig:xor-gadget-main}
  \end{figure}
  \ifbool{csub}{ 
    
    Given the above construction of $G$, our strategy for finding an
    alternating Hamiltonian path in $G$ is to start with the variable
    gadget for the special variable $\tau$, and then proceed to cover
    each of the clause gadgets in the order in which they are
    connected to this gadget.  In the course of covering a clause
    gadget, we also cover any variable gadgets attached to it that are
    not yet covered.  Since the dotted edges in all these
    gadgets have been removed, we need to take detours into the
    XOR-gadgets replacing them in order to complete the path.  The
    details of achieving this require considering appropriate partial
    Hamiltonian paths through the clause, variable and XOR-gadgets,
    and then using the gadget for special variable $\tau$ to stitch
    together these partial Hamiltonian paths to form an alternating
    Hamiltonian path in $G$.  The special orientations for the XOR and
    clause gadgets that we stipulated above are chosen to make this
    possible: see the full version for details.  }{
\subsection{Analyzing the reduction}
\label{sec:completing-reduction}
We now proceed to analyze the output of the reduction to complete the
proof of Theorem~\ref{thm:monomer-reduction-main}.  The use of
XOR-gadgets to enforce the ``pairing'' constraint as described above
introduces a factor of $2$ for each literal appearing in the clause,
and therefore the total weight of cycle covers after this step is
$2^{3\mu}s' \leq 2^{6\mu}$.  To get rid of the $-1$ weight edges in
the XOR-gadgets, 
we replace each such edge by a chain of $\kappa = 6\mu - 1$ vertices
with self-loops of weight $1$ and connecting edges of weight $2$.  We
call this final graph $G$.  Since the initial total weight of cycle
covers was at most $2^{6\mu}$, the weight of cycle covers in $G$ (and hence
the total-weight of all perfect matchings in \bip{G}), modulo
$2^{\kappa + 1} + 1$, is exactly $2^{3\mu}s'$.  Since all steps in the
construction of $\bip{G}$ starting from $\phi$ can be done in time
polynomial in the representation size of $\phi$, this proves parts 1
and 2 of Theorem\nobreakspace \ref {thm:monomer-reduction-main}.

We now proceed to prove part 3, that is, that $\bip{G}$ has a
Hamiltonian path.  We will use the alternating path notation described
above in order to keep our discussion in terms of the vertices and
edges of $G$, and we will call this representation of a Hamiltonian
path in $\bip{G}$ an \emph{alternating Hamiltonian path}.  In an
alternating Hamiltonian path, each vertex of $G$ is visited exactly
twice, and the length of the alternating path between the two visits
is odd.  This is equivalent to saying that all vertices must appear exactly
\emph{twice} in an alternating Hamiltonian path, with all vertices
except the first vertex in the path having one incoming \emph{forward}
edge, and one incoming \emph{backward} edge.
 
Our gadgets so far are designed to have alternating Hamiltonian paths
which can be pieced together to form an alternating Hamiltonian path
for $G$, and hence, we only need to list these paths and show
how to stitch them together.  We begin with alternating Hamiltonian
paths in the clause gadget.

  \begin{observation}
    \label{obv:clause-path}
    The clause gadget in Figure\nobreakspace \ref {fig:3-clause-gadget-m} has the
    alternating Hamiltonian path 
    \begin{displaymath}
      c \leftarrow 0 \rightarrow a \leftarrow c \rightarrow b \leftarrow a
      \rightarrow 0 \leftarrow b
    \end{displaymath}
    which uses all the dotted edges except the $b\rightarrow c$ dotted
    edge.  
  \end{observation}

  Recall that in  the construction of the reduction, we ensured that
  the new variable $\tau$ was associated with the $b\rightarrow c$
  dotted edge in each clause gadget.  This will be used to connect
  the above alternating Hamiltonian path in different clause gadgets
  via connections to the variable gadget for $\tau$ at the
  $b\rightarrow c$ edge.  Also, in the final construction, the
  dotted edges in the alternating Hamiltonian path will be
  replaced by detours into the associated XOR-gadget.  

  We now consider the XOR-gadget in Figure\nobreakspace \ref
  {fig:xor-gadget}.  It turns out that in some cases, we will need
  to traverse the XOR-gadget partially, so that a path enters at $a$
  via a backward edge, uses the $a\rightarrow d$ edge, and then leaves
  via a backward edge at $d$.  In order to cover the rest of the
  vertices, we will then need to construct an alternating path that
  enters at $a$ and leaves at $d$ via \emph{forward} edges, and covers
  all the vertices except $a$ and $d$ twice.  Another complication
  with the XOR-gadget is the presence of two $-1$ weight edges which
  need to be replaced with chains of vertices with self loops.
  However, this will not be a problem if we can ensure that both of
  the alternating paths described above use both the $-1$ weight
  edges, since an edge in an alternating path can always be replaced
  by a chain of vertices with self-loops.  We now show that, as we
  claimed above, our modified XOR-gadget satisfies all of these conditions.

  \begin{observation}
    \label{obv:xor-paths-1}
    The XOR-gadget in Figure\nobreakspace \ref {fig:xor-gadget} has the alternating
    Hamiltonian path
    \begin{displaymath}
      a \leftarrow b \rightarrow b \leftarrow a \rightarrow d \leftarrow c
    \rightarrow c \leftarrow d. 
    \end{displaymath}
    The gadget also has the following alternating path which enters at
    $a$ and leaves from $d$ using forward edges, but which does not
    otherwise visit these vertices:
  \begin{displaymath}
    a \leftarrow b \rightarrow b \leftarrow c \rightarrow c \leftarrow d.
  \end{displaymath}
  Moreover, both these paths use the $-1$ weight edges $b\rightarrow
  a$ and $c\rightarrow c$.
  \end{observation}

  \begin{remark}
    Since the XOR-gadget is connected to variable gadgets
    (except those for variable $\tau$) in $G$ via an outgoing edge at
    $a$ and an incoming edge at $d$, it will be possible to replace
    the $a\rightarrow d$ edge in the alternating Hamiltonian path
    above by a detour into the connected variable gadget when
    constructing an alternating Hamiltonian path in $G$.  Similarly,
    it will be possible to use the $a\rightarrow d$ edge as a
    replacement for the dotted edge in the variable gadget that was
    replaced by the XOR-gadget, at the cost of visiting the vertices
    $a$ and $d$ once.  The role of the second alternating path is
    to visit the remaining vertices in a XOR-gadget which has already
    been partially traversed in this way.
  \end{remark}
  \begin{remark}
    As pointed out above, it does not seem possible to include the two
    $-1$ weight edges in both the above alternating paths in Valiant's
    original construction.  This necessitated the construction of our
    new XOR-gadget in which the $-1$ weight edges are part of both the
    paths.  The edges in our construction are the same as those in
    Valiant's construction, but the weights have been chosen
    differently.
  \end{remark}
  
  We now consider the variable gadget shown in Figure\nobreakspace
  \ref {fig:var-gadget-1}.  We first work as if the dotted edges are
  present.  In this case, for any vertex $v$ in the gadget except the
  leftmost vertex, we can construct an alternating Hamiltonian path
  which covers all the vertices in the gadget, uses all the dotted
  edges except the one between $v$ and its predecessor and can be
  appended to an alternating path that enters via a forward edge at
  $v$ and leaves via a forward edge at $v$'s predecessor.  When the
  dotted edges are replaced by a XOR-gadget, this alternating path can
  still be traversed as described in the remarks following
  Observation\nobreakspace \ref{obv:xor-paths-1}, by instead following
  a forward edge to the $d$ vertex of the XOR-gadget, following the
  $a\rightarrow d$ edge in reverse, and then entering the variable
  gadget at the successor of $v$ via the outgoing edge at the $a$
  vertex of the XOR-gadget.  It is for this reason that we enforced
  above the condition that when a XOR-gadget is connected to a
  variable gadget for a variable other that $\tau$, it is oriented so
  that the incoming external edge at its $d$ vertex comes from the variable
  gadget.

  We now start constructing the alternating Hamiltonian path in $G$
  starting at the left-most vertex in the variable gadget for $\tau$.
  If all the dotted edges were present, this gadget is just a chain of
  vertices, and hence there is an alternating Hamiltonian path that
  covers all its vertices.  However, each dotted edge has been
  replaced by an outgoing edge to the $a$ vertex and an incoming edge
  from the $d$ vertex of a XOR-gadget.  Thus, instead of following the
  dotted edges, our alternating path will take a detour into the
  corresponding XOR-gadget, and after traversing several other
  vertices, return via its $d$ vertex to visit the other vertices in
  the variable gadget for $\tau$.  Thus, we need to show that
  these detours into the XOR-gadgets can be used to make the
  alternating path go through all the other vertices in $G$ twice
  while respecting the required parity constraints.

  We consider one such detour.  We suppose that the XOR-gadget in
  question connects to a clause gadget $C$ for the clause $\tau \lor
  v_1 \lor v_2$.  While following the alternating path for the
  XOR-gadget, we bypass the $a\rightarrow d$ edge of the XOR-gadget and
  instead take a detour into the $c$ vertex of $C$.  We then start
  following the alternating path in Observation\nobreakspace \ref
  {obv:clause-path}.  If the dotted edges were present, we would be
  able to complete an alternating path covering all vertices in $C$
  and then return via a forward edge from the $b$ vertex of $C$ into
  the $d$ vertex of the XOR-gadget.  We could then complete the
  alternating Hamiltonian path in the XOR-gadget, and return via a
  forward edge into the variable gadget for $\tau$.  However, since
  the dotted edges have been replaced by XOR-gadgets, we would need to
  take detours into the XOR-gadgets replacing them.  Suppose we are
  trying to replace the dotted edge $c\rightarrow a$, 
  corresponding to the literal $v_1$.  At this point there can be two
  cases:

  \paragraph{Case 1.} Suppose that the vertices of the variable gadget for
    $v_1$ have still not been covered by our growing alternating
    Hamiltonian path.  Consider the XOR-gadget $X$ replacing the
    $c\rightarrow a$ dotted edge of $C$.  We consider the alternating
    Hamiltonian path in Observation\nobreakspace \ref
    {obv:xor-paths-1} starting at the $a$ vertex of $X$.  We follow
    this path until we need to use the $a\rightarrow d$ edge.  At this
    point, we take a detour into the variable gadget for $v_1$ via a
    forward edge at vertex $a$ of $X$.  The vertex $u$ we connect to
    in the variable gadget cannot be a leftmost vertex, since its
    predecessor $u'$ is connected to vertex $d$ of $X$ via a
    $u'\rightarrow d$ edge.  As discussed above, we will therefore get
    an alternating Hamiltonian path for the vertex gadget which will
    leave the gadget through the $u'\rightarrow d$ edge (though this
    will end up using the $a\rightarrow d$ edges in all other
    XOR-gadgets corresponding to occurrences of $v_1$).  We can then
    complete the alternating Hamiltonian path for $X$, and this gives
    us an alternating path starting with a backward edge at vertex $a$
    of $X$, ending with a backward edge at vertex $d$ of $X$, and
    covering all vertices in $X$ and the variable gadget of $v_1$,
    while also using up the $a\rightarrow d$ edge in XOR-gadgets
    corresponding to all other occurrences of $v_1$.  We then use the
    forward edge from vertex $c$ of $C$ to vertex $a$ of $X$ and the
    forward edge from vertex $d$ of $X$ to vertex $a$ of $C$ to
    replace the dotted $c\rightarrow a$ edge by the above alternating
    path.
     \vspace{-1em}
  \paragraph{Case 2.} Suppose that the variable gadget for the vertex $v_1$
    has already been covered by our growing Hamiltonian path.  Then,
    as seen in Case 1, in the XOR-gadget $X$ corresponding to the
    $c\rightarrow a$ edge, the vertex $a$ has already been visited
    using a backward edge, while $d$ has already been visited via a
    forward edge.  Consider the second alternating path in
    Observation\nobreakspace \ref {obv:xor-paths-1}.  Traversing this
    alternating path from $a$ to $d$ will satisfy the remaining
    covering requirements for all the vertices in $X$.  Thus, as in
    Case 1 above, we can replace the $c\rightarrow a$ dotted edge in
    the alternating Hamiltonian path for $C$ by an edge from vertex
    $c$ of $C$ to vertex $a$ of $X$ and an edge from vertex $d$ of $X$
    to vertex $a$ of $X$.  As before, this ensures that the vertices
    of this XOR-gadget are covered while traversing the alternating
    Hamiltonian path for $C$.

    \vspace{1em}
    Observe that since each clause gadget is connected to the variable
    gadget for $\tau$, and since all other variable and XOR-gadgets are
    connected to at least one of the clause gadgets, the above
    alternating path eventually covers all of the individual gadgets.   
    This completes the proof for the existence of the alternating
    Hamiltonian path in $G$, and hence the proof of
    Theorem\nobreakspace \ref {thm:monomer-reduction-main}. \qed
}%

\section{Future Work}
\label{sec:open-problems}
This work leaves open the complexity of computing several other average
quantities; the most pertinent of which is $\avge{d}$,
the average size of cuts under the Ising measure.  The
obvious approach of attempting rational interpolation over $\beta$ via
an analog of our Theorem~\ref{thm:lee-yang-ext} does not directly work, since
Lee-Yang theorems do not hold in the same generality for the $\beta$
parameter.  A related problem is the complexity of computing the
susceptibility $\chi$ at \emph{fixed} values of $\lambda$ (in
particular, $\lambda = 1$), where, again, analyzing the partition
function and its derivatives as polynomials in $\beta$ may prove helpful.

Extensions of our results to the \emph{antiferromagnetic} Ising model
and the \emph{hard-core} (weighted independent sets)
model also remain open: again, our current approach would need to be
modified, since Lee-Yang theorems do not in general hold for
antiferromagnetic systems.  In particular, the best known analog of the
Lee-Yang theorem for the hard-core model, due to Chudnovsky and
Seymour~\cite{chudnovsky_roots_2007}, works only for claw-free
graphs.  Similarly, the complexity of computing averages in spin
systems with more than two spin values (such as the Potts model or proper
colorings) remains open.

Finally, we mention potential connections with the large literature on
stability preserving operators.  As indicated in
Section~\ref{sec:related}, this field is usually concerned with
operators that \emph{preserve} the region in which the zeros of a
polynomial lie.  Our extended Lee-Yang Theorem
(Theorem~\ref{thm:lee-yang-ext}) is an apparently rare example in which
the operator actually makes this region strictly smaller.   We
conjecture that there may be more applications of this phenomenon. 

\vspace{\baselineskip}
\noindent {\bf Acknowledgements.}
We thank Aris Anagnostopoulos for several fruitful discussions
\ifbool{csub}{and}{which led to the formulation of the problems
  considered in this paper.  We also thank} Marek Biskup, Christian
Borgs, Petter Br\"{a}nd\'{e}n, Charles Newman, David Ruelle, Alan
Sokal and Marc Thurley for helpful comments.

  %
  %
  %
  %
  %
  %
  %
  %
  %
  %
  %
  %
  %
  %
  %
  %
  %
  %
  %
%\bibliography{leeyang-j}

\begin{thebibliography}{10}

\bibitem{agrawal_primes_2004}
{\sc Agrawal, M., Kayal, N., and Saxena, N.}
\newblock {PRIMES} is in {P}.
\newblock {\em Ann. Math. 160}, 2 (Sept. 2004), 781--793.

\bibitem{arora09:_comput_compl}
{\sc Arora, S., and Barak, B.}
\newblock {\em Computational Complexity: A Modern Approach}.
\newblock Cambridge University Press, 2009.

\bibitem{asano_lee-yang_1970}
{\sc Asano, T.}
\newblock {Lee-Yang} theorem and the {Griffiths} inequality for the anisotropic
  {Heisenberg} ferromagnet.
\newblock {\em Phys. Rev. Lett. 24}, 25 (1970), 1409--1411.

\bibitem{biskup12}
{\sc Biskup, M.}, Feb. 2012.
\newblock Personal Communication.

\bibitem{bisborCKK04}
{\sc Biskup, M., Borgs, C., Chayes, J., Kleinwaks, L., and Koteck\'{y}, R.}
\newblock Partition function zeros at first-order phase transitions: A general
  analysis.
\newblock {\em Commun. Math. Phys. 251}, 1 (2004), 79--131.

\bibitem{biskup_partition_2004B}
{\sc Biskup, M., Borgs, C., Chayes, J., and Koteck\'{y}, R.}
\newblock Partition function zeros at first-order phase transitions:
  {Pirogov-Sinai} theory.
\newblock {\em J. Stat. Phys. 116}, 1 (2004), 97--155.

\bibitem{borcea_lee-yang_2009}
{\sc Borcea, J., and Br\"{a}nd\'{e}n, P.}
\newblock The {Lee-Yang} and {P\'{o}lya-Schur} programs. {I}. {Linear operators
  preserving stability}.
\newblock {\em Invent. Math. 177}, 3 (2009), 541--569.

\bibitem{borcea_polya-schur_2009}
{\sc Borcea, J., and Br\"{a}nd\'{e}n, P.}
\newblock {P\'{o}lya-Schur} master theorems for circular domains and their
  boundaries.
\newblock {\em Ann. Math. 170}, 1 (2009), 465--492.

\bibitem{brune_synthesis_1931}
{\sc Brune, O.}
\newblock {\em Synthesis of a finite two-terminal network whose driving-point
  impedance is a prescribed function of frequency}.
\newblock {Sc. D. Thesis}, MIT, 1931.

\bibitem{bul06}
{\sc Bulatov, A.}
\newblock A dichotomy theorem for constraint satisfaction problems on a
  3-element set.
\newblock {\em J. ACM 53\/} (2006), 66--120.

\bibitem{bulatov05}
{\sc Bulatov, A., and Grohe, M.}
\newblock The complexity of partition functions.
\newblock {\em Theor. Comput. Sci. 348}, 2-3 (2005), 148--186.

\bibitem{cai_complexity_2012}
{\sc Cai, J.-Y., and Chen, X.}
\newblock Complexity of counting {CSP} with complex weights.
\newblock In {\em Proc. ACM Symp. Theory Comput. (STOC) 2012\/} (2012), {ACM},
  pp.~909--920.

\bibitem{cai_graph_2010}
{\sc Cai, J.-Y., Chen, X., and Lu, P.}
\newblock Graph homomorphisms with complex values: A dichotomy theorem.
\newblock In {\em Proc. Int. Colloq. Autom. Lang. Program. (ICALP) 2010}.
  Springer Berlin/Heidelberg, 2010, pp.~275--286.

\bibitem{cai_dichotomy_2010}
{\sc Cai, J.-Y., and Kowalczyk, M.}
\newblock A dichotomy for $k$-regular graphs with $\{0, 1\}$-vertex assignments
  and real edge functions.
\newblock In {\em Proc. Conf. Theory Appl. Model. Comput. (TAMC) 2010\/}
  (2010), Springer Berlin/Heidelberg, pp.~328--339.

\bibitem{choe_homogeneous_2004}
{\sc Choe, Y.-B., Oxley, J.~G., Sokal, A.~D., and Wagner, D.~G.}
\newblock Homogeneous multivariate polynomials with the half-plane property.
\newblock {\em Adv. Appl. Math. 32}, 1--2 (2004), 88--187.

\bibitem{chudnovsky_roots_2007}
{\sc Chudnovsky, M., and Seymour, P.}
\newblock The roots of the independence polynomial of a clawfree graph.
\newblock {\em J. Comb. Theory Ser. B 97}, 3 (2007), 350--357.

\bibitem{clark_routh-hurwitz_1992}
{\sc Clark, R.}
\newblock The {Routh-Hurwitz} stability criterion, revisited.
\newblock {\em {IEEE} Control Syst. 12}, 3 (1992), 119--120.

\bibitem{cook_complexity_1971}
{\sc Cook, S.~A.}
\newblock The complexity of theorem-proving procedures.
\newblock In {\em Proc. ACM Symp. Theory Comput. (STOC) 1971\/} (1971), {ACM},
  pp.~151--158.

\bibitem{dyegre00}
{\sc Dyer, M.~E., and Greenhill, C.~S.}
\newblock The complexity of counting graph homomorphisms.
\newblock {\em Random Struct. Algorithms 17}, 3-4 (2000), 260--289.

\bibitem{edmonds_paths_1965}
{\sc Edmonds, J.}
\newblock Paths, trees, and flowers.
\newblock {\em Canad. J. Math. 17}, 0 (Jan. 1965), 449--467.

\bibitem{edmonds67:_system_distin_repres_linear_algeb}
{\sc Edmonds, J.}
\newblock Systems of distinct representatitives and linear algebra.
\newblock {\em J. Res. Natl. Bur. Stand. Sect. B 71B}, 4 (1967), 241--245.

\bibitem{golgrojerthu09}
{\sc Goldberg, L.~A., Grohe, M., Jerrum, M., and Thurley, M.}
\newblock A complexity dichotomy for partition functions with mixed signs.
\newblock {\em {SIAM} J. Comput. 39}, 7 (2010), 3336--3402.

\bibitem{goljerpat03}
{\sc Goldberg, L.~A., Jerrum, M., and Paterson, M.}
\newblock The computational complexity of two-state spin systems.
\newblock {\em Random Struct. Algorithms 23\/} (2003), 133--154.

\bibitem{griffiths_correlations_1967}
{\sc Griffiths, R.~B.}
\newblock Correlations in {Ising} ferromagnets. {I}.
\newblock {\em J. Math. Phys. 8}, 3 (1967), 478--483.

\bibitem{heilmann_monomers_1970}
{\sc Heilmann, O.~J., and Lieb, E.~H.}
\newblock Monomers and dimers.
\newblock {\em Phys. Rev. Lett. 24}, 25 (1970), 1412--1414.

\bibitem{heilmann_theory_1972}
{\sc Heilmann, O.~J., and Lieb, E.~H.}
\newblock Theory of monomer-dimer systems.
\newblock {\em Commun. Math. Phys. 25}, 3 (1972), 190--232.

\bibitem{jervalvaz86}
{\sc Jerrum, M., Valiant, L.~G., and Vazirani, V.~V.}
\newblock Random generation of combinatorial structures from a uniform
  distribution.
\newblock {\em Theoret. Comput. Sci. 43\/} (1986), 169--188.

\bibitem{citeulike:3082225}
{\sc Karp, R.~M.}
\newblock Reducibility among combinatorial problems.
\newblock In {\em Proc. Complexity of Computer Computations\/} (1972), Plenum
  Press, pp.~85--103.

\bibitem{laguerre82:_sur}
{\sc Laguerre, E.}
\newblock Sur les fonctions du genre z\'{e}ro et du genre un.
\newblock {\em C. R. Acad. Sci. 98\/} (1882).

\bibitem{leeyan52b}
{\sc Lee, T.~D., and Yang, C.~N.}
\newblock Statistical theory of equations of state and phase transitions. {II}.
  {Lattice} gas and {Ising} model.
\newblock {\em Phys. Rev. 87}, 3 (1952), 410--419.

\bibitem{levin73:_univer_russian}
{\sc Levin, L.~A.}
\newblock Universal search problems.
\newblock {\em Probl. Inf. Transm. 9}, 3 (1973), 265--266.

\bibitem{macon_existence_1962}
{\sc Macon, N., and Dupree, D.~E.}
\newblock Existence and uniqueness of interpolating rational functions.
\newblock {\em Amer. Math. Monthly 69}, 8 (1962), 751--759.

\bibitem{newman_zeros_1974}
{\sc Newman, C.~M.}
\newblock {Zeros of the partition function for generalized Ising systems}.
\newblock {\em Commun. Pure Appl. Math. 27}, 2 (1974), 143--159.

\bibitem{newman_ghs_1991}
{\sc Newman, C.~M.}
\newblock The {GHS} inequality and the {Riemann} hypothesis.
\newblock {\em Constr. Approx. 7}, 1 (1991), 389--399.

\bibitem{pap94}
{\sc Papadimitriou, C.~H.}
\newblock {\em Computational Complexity}.
\newblock Addison-Wesley, 1994.

\bibitem{polya14:_uber_arten_faktor_theor_gleic}
{\sc P\'{o}lya, G., and Schur, I.}
\newblock \"{U}ber zwei {Arten von Faktorenfolgen in der Theorie der
  Algebraischen Gleichungen}.
\newblock {\em J. Reine Angew. Math. 144\/} (1914), 89--113.

\bibitem{ruelle83:_statis_mechan}
{\sc Ruelle, D.}
\newblock {\em Statistical Mechanics}.
\newblock W. A. Benjamin Inc., 1983.

\bibitem{Sly2010CompTransition}
{\sc Sly, A.}
\newblock Computational transition at the uniqueness threshold.
\newblock In {\em Proc. IEEE Symp. Found. Comput. Sci. (FOCS) 2010\/} (2010),
  IEEE Computer Society, pp.~287--296.

\bibitem{sly12}
{\sc Sly, A., and Sun, N.}
\newblock The computational hardness of counting in two-spin models on
  d-regular graphs.
\newblock In {\em Proc. IEEE Symp. Found. Comput. Sci. (FOCS) 2012\/} (2012),
  pp.~361--369.

\bibitem{suzuki_zeros_1971}
{\sc Suzuki, M., and Fisher, M.~E.}
\newblock Zeros of the partition function for the {Heisenberg}, ferroelectric,
  and general {Ising} models.
\newblock {\em J. Math. Phys. 12}, 2 (1971), 235--246.

\bibitem{val79a}
{\sc Valiant, L.~G.}
\newblock The complexity of computing the permanent.
\newblock {\em Theor. Comput. Sci. 8\/} (1979), 189--201.

\bibitem{val79b}
{\sc Valiant, L.~G.}
\newblock The complexity of enumeration and reliability problems.
\newblock {\em SIAM J. Comput. 8}, 3 (1979), 410--421.

\bibitem{Weitz06CountUptoThreshold}
{\sc Weitz, D.}
\newblock Counting independent sets up to the tree threshold.
\newblock In {\em Proc. ACM Symp. Theory Comput. (STOC) 2006\/} (2006), ACM,
  pp.~140--149.

\bibitem{leeyan52}
{\sc Yang, C.~N., and Lee, T.~D.}
\newblock Statistical theory of equations of state and phase transitions. {I.
  Theory} of condensation.
\newblock {\em Phys. Rev. 87}, 3 (1952), 404--409.

\end{thebibliography}
%\bibliographystyle{acm}

\appendix
%\noindent{\LARGE \textbf{Appendix}}\nopagebreak

\section{An overview of complexity theoretic terminology}
\label{sec:an-overv-compl}

For the benefit of readers who may not be familiar with terminology
from complexity theory, we provide here an overview of the relevant
definitions.   For more details, we refer to the textbook of Arora and
Barak~\cite{arora09:_comput_compl}.  

The conventional notion of efficiency in complexity theory is
computability in time polynomial in the size of the input.  For
example, the problem of finding the shortest path between two
specified vertices in a graph (presented as an adjacency matrix) can be
solved in ``polynomial time'', e.g., using Djikstra's algorithm.  This
notion of efficiency is captured by the
complexity class P.  For technical reasons we consider \emph{decision
  problems}, where the answer is either `YES' or `NO'.  For example,
the decision version of the shortest path problem would be to determine,
given an integer $\ell$, if there is a path of length at most $\ell$
between the two specified vertices.  A decision problem is said to be
in the class P if there is an algorithm for it that runs in time
polynomial in the size of the input.  Examples of problems known to be
in $P$ include deciding if a graph has an Eulerian tour, deciding
whether an integer (presented in binary) is
prime~\cite{agrawal_primes_2004}, etc.

The class NP is the class of decision problems for which a polynomial
time algorithm can verify the correctness of a `YES' answer given a
\emph{certificate} of correctness.   For example, we consider SAT,
the problem of deciding whether a given Boolean formula has a satisfying
assignment.   While this problem is conjectured to not belong in the
class P, it is clearly in NP: the certificate for a formula being
satisfiable is just a satisfying assignment, if any, of the formula.
Formally, a problem is said to be in NP if there is an algorithm
$\mathcal{A}$ and a polynomial $p$ such that for any instance $I$ of
the problem:
\begin{enumerate}
\item If $I$ is a `YES' instance, then there exists a certificate $C$
  such that on input $(I,C)$, $\mathcal{A}$ runs in time at most
  $p(|I|)$ and outputs `YES'.
\item If $I$ is a `NO' instance, then for any certificate $C$, on
  input $(I,C)$, $\mathcal{A}$ runs in time at most $p(|I|)$ and outputs
  `NO'. 
\end{enumerate}

Another example of a problem in NP is deciding whether a graph $G$ has
a Hamiltonian cycle (where the certificate is a Hamiltonian cycle, if
any, in the graph).  Both of the above examples are also \emph{NP-hard}, that
is, if there is a polynomial time algorithm for either of them, then
there are polynomial time algorithms for all problems in NP.  It is
clear that P $\subseteq$ NP, but it is a long standing conjecture in
complexity theory that P $\subsetneq$ NP.  NP-hardness has long been
the standard notion of intractability for decision problems, since the
pioneering work of Cook~\cite{cook_complexity_1971},
Levin~\cite{levin73:_univer_russian} and
Karp~\cite{citeulike:3082225}.  

The class \#P, introduced by Valiant~\cite{val79a}, is a counting
analog of NP: any problem in \#P involves counting the number of
certificates (possibly zero) that would make an algorithm for an
NP-problem accept.  An example of a problem in \#P is counting the
number of satisfying assignments of a given SAT formula.  As in the
case of NP-hardness, a counting problem is called \emph{\#P-hard} if a
polynomial time algorithm for the problem implies the existence of
polynomial time algorithms for all problems in \#P.  

Counting versions of most NP-complete problems (such as counting the
number of satisfying assignments of a Boolean formula) can be shown to
be \#P-hard.  However, there are decision problems in P whose counting
analogs are still \#P-hard, and it is this fact that makes the theory
of \#P-hardness a non-trivial extension of the theory of NP-hardness.  A famous example is counting the number
of dimer coverings (perfect matchings) of a given graph:
Edmonds~\cite{edmonds_paths_1965} showed that there is a polynomial
time algorithm for checking whether a graph has a dimer
covering;\footnote{Edmond's result was, in fact, a crucial step towards
  the definition of the class P.} however, counting the number of
such coverings in general graphs is \#P-hard~\cite{val79a}.
\#P-hardness has since become the standard notion of intractability
for counting problems and for the computation of partition functions
in statistical physics:
indeed, almost all interesting classes of partition functions are
known to be \#P-hard to compute (see, e.g., \cite{bulatov05}).

\section{Proof of Theorem\nobreakspace \ref {thm:newman}}
\label{sec:proof-crefthm:newman}
As indicated earlier, Theorem\nobreakspace \ref {thm:newman} is an easy corollary of the
following result of Newman~\cite{newman_zeros_1974} (restated in our
notation).
\begin{theorem}[{\cite[Theorem 3.2 and eq. 3.5]{newman_zeros_1974}}]
\label{thm:newman-orig}
Consider the ferromagnetic Ising model on a graph $G = (V,E)$ on $n$ vertices,
where we allow the edge potentials also to be variable, with the
condition that the edge potential $\beta_{uv}$ on every edge $uv$
satisfies $0 < \beta_{uv} < 1$.  Let $\inp{z_v}_{v\in V}$ be a
collection of complex vertex activities such that $\abs{z_v} > 1$ for
all $v\in V$.  Then,
\begin{displaymath}
  \Re\inp{M(G,\inp{\beta_{uv}}_{uv\in E}, \inp{z}_{v\in V})} > \frac{n}{2}.
\end{displaymath}
\end{theorem}

We now proceed with the proof of Theorem\nobreakspace \ref {thm:newman}.
\begin{proof}[Proof of Theorem\nobreakspace \ref {thm:newman}]
  Let $w$ be any weight assignment (not necessarily legal) of positive
  integral weights to the vertices of $G$.  Consider the graph $H$
  obtained from $G$ by appending to each vertex $v$ of $G$ a chain
  $C_v$ of $w(v) - 1$ vertices.  Further, we let the edge potential be
  $\beta$ on all edges of $H$ which were present in $G$, and $0 <
  \gamma < 1$ on all the edges which are either part of some $C_v$, or connect
  a vertex $v$ to its associated chain $C_v$.  Let $\inp{y_v}_{v\in
    V}$ be a set of vertex activities on $V$.  Henceforth, we will
  drop the subscript and refer to this set of activities as
  $\inp{y_v}$.  With a slight abuse of notation, we also denote by
  $\inp{y_v}$ the collection of activities on $H$ such that for any vertex $x$
  in $H$ such that $x \in C_v$ for some $v \in V$, we have $y_x = y_v$.

  Now consider any collection of activities such that $\abs{y_v} > 1$ for
  all $v \in V$.  From Theorem\nobreakspace \ref {thm:newman-orig}, we get that
  for any $\gamma \in \inp{0,1}$,
  \begin{equation}
    \Re\inp{M(H, \inb{\beta,\gamma}, \inp{y_v})} >
    \frac{n}{2}.\label{eq:26}
  \end{equation}
  Since $\abs{y_v} > 1$ for all $v$, the Lee-Yang
  theorem\footnote{Although we stated Theorem\nobreakspace \ref {thm:lee-yang-2} only for
    uniform edge potentials, Asano's proof~\cite{asano_lee-yang_1970}
    in fact supports our current conclusion with variable edge
    potentials and a possibly disconnected graph.} implies that both
  $Z_I(H, \inb{\beta,\gamma}, \inp{y_v})$ and $Z_w(G, \beta,
  \inp{y_v})$ are non-zero (when $0 < \gamma < 1$).  We can
  therefore take the limit $\gamma \rightarrow 0$ in eq.\nobreakspace \textup {(\ref {eq:26})} to
  get
  \begin{equation}
    \label{eq:27}
   \frac{n}{2} \leq \lim_{\gamma\rightarrow 0} \Re\inp{M(H,
      \inb{\beta,\gamma}, \inp{y_v})}  = \Re\inp{\frac{Z_w'(G,\beta,\inp{y_v})}{Z_w(G,\beta,\inp{y_v})}} =  \Re\inp{M(G, \beta, \inp{y_v})}.
  \end{equation}
  We now take a sequence $\inp{\inp{y^\ell_v}}_{\ell=1}^{\infty}$ of activity
  assignments such that  $\abs{y^\ell_v} > 1$ for all $\ell$ and $v$ and
  such that $\lim y_v^\ell = z_v$.  Since we assume in the hypothesis of
  the theorem that $Z_w(G, \beta, (z_v)_{v\in V}) \neq 0$, we can take the limit $\ell
  \rightarrow \infty$ in eq.\nobreakspace \textup {(\ref {eq:27})} to get 
  \begin{displaymath}
    \frac{n}{2} \leq \lim_{\ell\rightarrow\infty} \Re\inp{M\inp{G, \beta,
      \inp{y^\ell_v}}} =   \Re\inp{M(G)},
  \end{displaymath}
  which completes the proof. 
\end{proof}

\section{Hardness with a smaller blowup in degree}
\label{sec:bdd-degree-hard}
Recall that, in our proofs of Theorems~\ref{thm:ising-main} and
\ref{thm:monomer-dimer-main} in Sections~\ref{sec:hardn-comp-magn} and
\ref{sec:hardn-comp-aver}, we realized different values of $\lambda$
required for the interpolation by attaching $k$ extra vertices to each
vertex of $G$.  This necessarily entails a large increase in the degree of
$G$. In this section, we give an alternative way of realizing
different values of $\lambda$ which entails an increase in degree of
exactly one, and which therefore allows us to complete the proofs of
the stronger, degree-bounded version of Theorems \ref{thm:ising-main}
and \ref{thm:monomer-dimer-main}.

We denote by $P_k$ a path of $k$ vertices. Let $p_k^+$ (respectively,
$p_k^-$) be the partition function $Z_I(P_k, \beta, \lambda)$
restricted to configurations in which the leftmost vertex of $P_k$ is
fixed to spin `$+$' (respectively, `$-$').  We also set $r_k \defeq
\frac{p_k^+}{p_k^-}$.  Similarly, we denote by $y_k$ the partition
function $Z_M(P_k,\lambda)$, where we assume that all edges in $P_k$
have weight one, and suppress the dependence on edge weights for
clarity of notation.
Notice that $p_1^- = 1$ and $p_1^+ = r_1 = y_1 =
\lambda$.  We further define $y_0 = 1$. The following recurrence
relations show that, for fixed $\beta$ and $\lambda$, $p_k^+, p_k^-$,
$r_k$ and $y_k$ can be computed in time polynomial in $k$:
\begin{align}
  p_k^+ &= \lambda(\beta p_{k-1}^- + p_{k-1}^+);\label{eq:14}\\
  p_k^- &= \beta p_{k-1}^+ + p_{k-1}^-;\label{eq:15}\\
  r_k &= \lambda\frac{\beta + r_{k-1}}{1 + \beta r_{k-1}};\label{eq:16}\\
  y_k &= \lambda y_{k-1} + y_{k-2}\label{eq:28}.
\end{align}
Notice that $p_k^+, p_k^-$, $r_k^+$ and $y_{k}$ are all functions of
$\lambda$.  We note that values of their derivatives with respect to
$\lambda$ can also be computed in time polynomial in $k$ via the
following recurrence relations: $\dot{p}_1^- = \dot{y}_0 = 0, \dot{p}_1^+ =
\dot{r}_1 = \dot{y}_1 = 1$, and
\begin{align}
  \dot{p}_k^+ &= p_k^+/\lambda + \lambda(\beta \dot{p}_{k-1}^- + \dot{p}_{k-1}^+);\label{eq:17}\\
  \dot{p}_k^- &= \beta \dot{p}_{k-1}^+ + \dot{p}_{k-1}^-;\label{eq:18}\\
  \dot{r}_k &= \frac{\dot{p}_k^+p_k^- -
    p_k^+\dot{p}_k^-}{\inp{p_k^-}^2};\label{eq:19}\\
  \dot{y_k} &= y_{k-1} + \lambda\dot{y}_{k-1} +
  \dot{y}_{k-2}.\label{eq:29}
\end{align}
Here, we use the dot notation for the derivative with respect to
$\lambda$.  Using a simple induction, one can also show that when
$\beta < 1$, $\dot{r}_k > 0$ for all $k$.

Now consider a connected graph $G$.  For $k \geq 1$, we define $G(k)$
as the graph obtained by attaching to each vertex $v$ of $G$ a
different instance of the path $P_k$, such that $v$ is connected to
the ``leftmost'' vertex of $P_k$ via an edge.  Notice that the maximum
degree of $G(k)$ is one more than the maximum degree of $G$.
\noindent We first consider the Ising model on the graphs $G(k)$.  We have
\begin{equation}
  \label{eq:20}
  Z_I(G(k), \beta, \lambda) = (p_{k+1}^-)^nZ_I(G, \beta, \lambda_k),
\end{equation}
where $n$ is the number of vertices in $G$ and $\lambda_k = r_{k+1}$.
Notice that when $0 < \beta < 1$, the sequence $\lambda_k$ is strictly
increasing and greater than $1$ (respectively, strictly decreasing and
less than $1$) when $\lambda > 1$ (respectively, when $\lambda < 1$):
this follows from the observation that the right hand side of the
recurrence (\ref{eq:16}) is a strictly increasing function of
$r_{k-1}$, and that $r_2 > r_1$ (respectively, $r_2 < r_1$) when
$\lambda > 1$ (respectively, when $\lambda < 1$).  We also have
\begin{equation}
  M(G(k), \beta, \lambda) = \frac{n\lambda\dot{p}_{k+1}^-}{p_{k+1}^-} + \frac{\lambda\dot{r}_{k+1}}{r_{k+1}}M(G,\beta, \lambda_k). \label{eq:21}
\end{equation}
\noindent We now complete the proof of Theorem\nobreakspace \ref
{thm:ising-main}.
\begin{proof}[Proof of Theorem\nobreakspace \ref {thm:ising-main}]
  As in the partial proof in Section~\ref{sec:hardn-comp-magn}, we
  assume $\lambda > 1$ (since the case $\lambda < 1$ is symmetrical)
  and suppose that we have a polynomial time algorithm $\mathcal{A}$
  which, given a connected graph $G$ of maximum degree at most $\Delta
  \geq 4$, outputs the mean magnetization $M(G,\beta,\lambda)$ in
  polynomial time.

  Now consider any connected graph $G$ of maximum degree at most
  $\Delta - 1 \geq 3$.  As shown in the partial proof in
  Section~\ref{sec:hardn-comp-magn}, Theorem~\ref{thm:lee-yang-ext}
  implies that
  if we can efficiently evaluate $M(G, \beta, z)$ at $2n + 2$ distinct
  values of $z$ using our hypothetical algorithm ${\cal A}$, we can
  uniquely determine the coefficients of $Z_I(G,\beta,z)$, and hence
  also $Z_I(G,\beta, 1)$, in polynomial time.  In view of
  Theorem\nobreakspace \ref {thm:ising-partition-hard}, this would
  imply that the problem of computing the mean magnetization in graphs
  of maximum degree at most $\Delta$ for parameter values $\beta$
  and $\lambda$ is \#P-hard.

  In order to evaluate $M(G, \beta, z)$ at $2n + 2$ distinct values,
  we evaluate $M(G(k),\allowbreak \beta,\allowbreak \lambda)$ for $1 \leq k \leq 2n
  + 2$ using our hypothetical algorithm ${\cal A}$.  Notice that this
  can be done since the construction of the $G(k)$ (as given in this
  section) implies that they have maximum degrees which are at most one
  larger than the maximum degree of $G$.  Using eqs.\nobreakspace
  \textup {(\ref {eq:14})} to\nobreakspace \textup {(\ref {eq:19})}
  and\nobreakspace \textup {(\ref {eq:21})}, and the fact that
  $\dot{r}_k> 0$ for all $k$, we can then determine $M(G,\beta,
  \lambda_k)$ in polynomial time.  Since $\lambda_k$ is a strictly
  increasing sequence, these evaluations are at distinct points, and
  hence the reduction is complete.
\end{proof}

\noindent We now consider the monomer-dimer model on the graphs
$G(k)$. We have
\begin{equation}
  \label{eq:30}
  Z_M(G(k), \lambda) = y_k^nZ_M(G, \lambda_k),
\end{equation}
where $n$ is the number of vertices in $G$ and $\lambda_k =
y_{k+1}/y_{k}$.  We also have
\begin{equation}
  U(G(k), \lambda) = n\lambda t_k + \lambda(t_{k+1}-t_{k}) U(G,\lambda_k)\label{eq:31},
\end{equation}
where $t_k = \dot{y}_k/y_{k}$.  It turns out that the sequence
$\inp{\lambda_{2k}}_{k \geq 0}$ is strictly increasing and hence
consists of distinct values, and further that $t_{2k+1} - t_{2k} > 0$
for all $k$.  This follows easily from the following explicit
solutions for the $y_k$ and the $\lambda_k$:
\begin{displaymath}
  y_k = \frac{\xi^{k+1}-\eta^{k+1}}{\xi-\eta};\qquad\lambda_k = \frac{\xi^{k+2} - \eta^{k+2}}{\xi^{k+1} - \eta^{k+1}}, 
\end{displaymath}
where
\begin{displaymath}
  \xi = \frac{1}{2}\inp{\lambda + \sqrt{\lambda^2 + 4}} > 0; \qquad 
  \eta = \frac{1}{2}\inp{\lambda - \sqrt{\lambda^2 + 4}} < 0.
\end{displaymath}
Notice that $t_{k+1} - t_{k} > 0$ for even $k$ implies that for such
$k$, we can determine $U(G,\lambda_k)$ given $U(G(k),\lambda)$, using
equations (\ref{eq:28}), (\ref{eq:29}) and (\ref{eq:31}).  We can now
complete the proof of 
Theorem~\ref{thm:monomer-dimer-main} for the bounded degree case.
\begin{proof}[Proof of Theorem~\ref{thm:monomer-dimer-main}]
  As in the partial proof in Section~\ref{sec:hardn-comp-aver}, we fix
  any $\lambda > 0$, and suppose that there exists a polynomial time
  algorithm $\mathcal{B}$ which, given a connected graph $H$ with edge
  weights in the set $\inb{1, 2, 3}$, and of maximum degree at most $\Delta
    \geq 5$, outputs $D(H,\lambda)$ (recall that we are
  suppressing explicit dependence on the edge weights for clarity of
  notation).  Given a {\sc Monotone 2-SAT} formula $\phi$, we then
  produce the graph $G = \mathcal{A}(\phi)$ in polynomial time.
  Notice that in the construction of $G$ as given in
  Section~\ref{sec:overv-our-reduct}, each vertex has degree at most $4
  \leq \Delta - 1$: this corresponds to the maximum of the in-degrees
  and the out-degrees over all vertices in the directed version of the
  reduction.

  As argued in the partial proof in Section~\ref{sec:hardn-comp-aver},
  Theorem\nobreakspace \ref {thm:heilmann-lieb} and the existence of a
  Hamiltonian path in $G$ together imply that if we could use
  algorithm ${\cal B}$ to evaluate $D(G,z)$ (and hence $U(G, z)$) at
  $2n + 2$ different values of $z$, then we can determine the number
  of satisfying assignments of $\phi$ in polynomial time.  This would
  in turn imply that computing $D(G,\lambda)$ for graphs of maximum
  degree at most $\Delta$ is \#P-hard.

  As before, in order to realize other values of $\lambda$, we
  consider the graphs $G(k)$ (as described in this section), for $k =
  0, 2, 4, \ldots 4n + 4$.  Notice that the maximum degree of $G(k)$ is
  one more than that of $G$, and hence is at most $\Delta$.
  Further, as argued in the remarks following eqs. (\ref{eq:30}) and
  (\ref{eq:31}), these choices of $k$ ensure that the values
  $\lambda_k$ are distinct, and that $U(G,\lambda_k)$ can be easily
  determined from $U(G(k), \lambda)$.  We can therefore determine
  $U(G,z)$ at $2n + 2$ different values of $z$ by running
  $\mathcal{B}$ on the $G(k)$, as required.
\end{proof}

\section{Hardness for general two state feromagnetic spin systems and\\
  planar graphs}
\label{sec:hardness-general-two}

We now show how to extend our results to general two-state
ferromagnetic spin systems.  Recall that a general two-state spin
system~\cite{goljerpat03} is parametrized by a $(+,+)$ edge potential
$\alpha_1$, a $(-,-)$ edge potential $\alpha_2$, and a vertex
activity $\lambda$.  As before, given a graph $G = (V,E)$, we define
a probability distribution over the set of configurations $\sigma : V
\rightarrow \inb{+,-}$ via the weights $w_S(\sigma)$ given by 
\begin{displaymath}
  w_S(\sigma) = \lambda^{p(\sigma)}\alpha_1^{e_+(\sigma)}\alpha_2^{e_-(\sigma)}, 
\end{displaymath}
where $e_+(\sigma)$ (respectively, $e_-(\sigma)$) denotes the number
of edges with `$+$' (respectively, `$-$') spin on both end-points,
while $p(\sigma)$ denotes the number of vertices with $+$ spin.  The
partition function $Z_S(G, \alpha_1, \alpha_2, \lambda)$ and the
magnetization $M_S(G,\alpha_1, \alpha_2, \lambda)$ are given by
\begin{align*}
  Z_S(G,\alpha_1, \alpha_2, \lambda) &\defeq \sum_{\sigma \in
    \inb{+,-}^V}w_S(\sigma);\\
  M_S(G,\alpha_1, \alpha_2, \lambda) := \avge{p} &= \frac{\sum_{\sigma}p(\sigma)w_S(\sigma)}{Z_S(G,\alpha_1,\alpha_2,\lambda)}.
\end{align*}
\begin{remark}
  The Ising model corresponds to the special case $\alpha_1 = \alpha_2
  = \beta$.
\end{remark}
It is well known that general two-state spin systems can be
represented in terms of an Ising model in which the activity at each
vertex depends upon the degree of the vertex~\cite{goljerpat03}.  In
particular, if $G$ is a $\Delta$-regular graph then all vertex
activities in the equivalent Ising model are the same, and one has
\begin{equation}
  w_S(\sigma) = \alpha_2^{|E|}w_I(\sigma)\label{eq:41}
\end{equation}
where the Ising model has an edge potential $\beta =
1/\sqrt{\alpha_1\alpha_2}$ and a vertex activity $\lambda' =
\lambda(\alpha_1/\alpha_2)^{\Delta/2}$.  A two-spin system is called
\emph{ferromagnetic} if the above translation produces a ferromagnetic
Ising model, that is, when $\alpha_1\alpha_2 \geq 1$.

However, the above translation does not allow us to directly translate
our hardness result for the ferromagnetic Ising model, since our
results were not derived for regular graphs.  We will instead do a
reduction similar to the ones done in our earlier proof, but starting
from the following somewhat stronger hardness result for the partition
function.

\begin{remark}
  In this section, we allow graphs to have parallel edges and
  self-loops.  In computing the degree of a vertex, each self-loop is
  counted twice (since it is incident twice on the vertex) and each
  parallel edge is counted separately.  As observed in the remark
  following the proof of Theorem~\ref{thm:wsglp} in
  Section~\ref{sec:an-extended-lee}, our extended Lee-Yang theorem
  (Theorem~\ref{thm:lee-yang-ext}) holds also in this setting.
\end{remark}
\begin{theorem}[{\cite[Theorem 1]{cai_dichotomy_2010}}] 
 \label{thm:hardness-regular}
  Fix $\alpha_1, \alpha_2 > 0$ with $\alpha_1\alpha_2 > 1$ and $\Delta
  \geq 3$.  The problem of computing the partition function $Z_S(G,
  \alpha_1, \alpha_2, 1)$ on $\Delta$-regular graphs is \#P-hard.
\end{theorem}
We will then prove the following theorem. We will also show later in this
section that the theorem can be strengthened so that the \#P-hardness
holds even when the input is restricted to planar graphs.
\begin{theorem}
  \label{thm:hardness-general-two}
  Fix $\alpha_1, \alpha_2, \lambda > 0$ and $\Delta \geq 4$ such that
  $\alpha_1\alpha_2 > 1$ .  The problem of computing the magnetization
  $M_S(G, \alpha_1, \alpha_2, \lambda)$ on connected graphs of degree
  at most $\Delta$ is \#P-hard, except when $\alpha_1 = \alpha_2$ and
  $\lambda = 1$, in which case it can be solved in polynomial time.
\end{theorem}
\begin{remark}
  Notice that when $\alpha_1\alpha_2 =1$, the problem reduces to the
  case of a graph consisting of isolated vertices, and hence can be
  solved in polynomial time.  Similarly, in the case $\alpha_1 =
  \alpha_2$ and $\lambda = 1$, the two spins are symmetric, and the
  magnetization is therefore $n/2$, where $n$ is the number of
  vertices in $G$.
\end{remark}
Before proceeding with the proof of
Theorem~\ref{thm:hardness-general-two}, we will need to analyze the
model on graphs $G(k)$ defined in Appendix~\ref{sec:bdd-degree-hard}.
As before, we begin by analyzing the model on the path $P(k)$.  We
denote by $p_k^{+}$ (respectively, $p_k^-$) the partition function
$Z_S(P_k, \alpha_1, \alpha_1, \lambda)$ restricted to configurations
in which the leftmost vertex is fixed to be `$+$' (respectively,
`$-$').  We also define the ratio $r_k = p_k^+/p_k^-$.  Similarly, we
denote by $m_k^+$ (respectively, $m_k^-$) the average magnetization of
the path $P_k$ conditioned on the leftmost vertex being fixed to `$+$'
(respectively, `$-$').  We have $p_1^+ = r_1 = \lambda, p_1^- = 1$ and
$m_1^+ = 1, m_1^- = 0$, and the following recurrences for $k \geq 1$:
\begin{align}
  p_{k}^+ &= \lambda(\alpha_1p_{k-1}^+ + p_{k-1}^-)\label{eq:34}\\
  p_{k}^- &= \alpha_2p_{k-1}^- + p_{k-1}^+\label{eq:35}\\
  r_{k} &= \lambda\frac{\alpha_1 r_{k-1} + 1}{\alpha_2 + r_{k-1}}\label{eq:38}\\
  m_{k}^+ &= 1 + \frac{\alpha_1m_{k-1}^+p_{k-1}^+ +
    m_{k-1}^-p_{k-1}^-}{\alpha_1p_{k-1}^+ + p_{k-1}^-}\label{eq:36}\\
  m_{k}^- &= \frac{\alpha_2m_{k-1}^-p_{k-1}^- +
    m_{k-1}^+p_{k-1}^+}{\alpha_2p_{k-1}^- + p_{k-1}^+}\label{eq:37}
\end{align}
Under the condition $\alpha_1\alpha_2 > 1$, one can prove using a
simple induction that for all $k \geq 1$, $m_k^+ - m_k^- > 0$, and
that when $(\alpha_1 - 1)\lambda - (\alpha_2 - 1) > 0$ (respectively,
when $(\alpha_1 - 1)\lambda - (\alpha_2 - 1) < 0$), the $r_k$ form a
strictly increasing (respectively, strictly decreasing) sequence, and
hence are all distinct.

For reasons that will become clear shortly, we need the $r_k$ to be
distinct, and hence we will need to handle the remaining case
$(\alpha_1 - 1)\lambda - (\alpha_2 - 1) = 0$ specially.  We observe
that unless $\alpha_1=\alpha_2 = 1$, or $\alpha_1 = \alpha_2$ and
$\lambda = 1$, both of which are excluded in the hypothesis of the
theorem, we cannot have both $(\alpha_1 - 1)\lambda - (\alpha_2 - 1) =
0$ and $(\alpha_1^2 - 1)\lambda - (\alpha_2^2 - 1) = 0$.  To take
advantage of this, we will modify $P(k)$ by replacing each edge in
$P(k)$ by two parallel edges.  We call the resulting graph $P(k)'$,
and again define the quantities $p_k^+, p_k^-. r_k, m_k^+$ and $m_k^-$
by recursion on $P(k)'$.  Notice that for $k = 1$, these quantities
are the same as those for $P(k)$; however for $k \geq 2$, we now need
to modify the recurrences above by replacing $\alpha_1$ and $\alpha_2$
by $\alpha_1^2$ and $\alpha_2^2$ respectively.  As before, we have
$m_k^+ - m_k^- > 0$ for all $k \geq 1$.  Further, by our observation,
the $r_k$ form a strictly monotone sequence.  Thus, in the case
$(\alpha_1 - 1)\lambda - (\alpha_2 - 1) = 0$, we redefine $G(k)$ to
use the paths $P(k)'$ in place of $P(k)$.  In what follows, we will
assume that $G(k)$ are appropriately defined taking into account the
values of $\lambda, \alpha_1$ and $\alpha_2$, and will not explicitly
keep track of the above modification.  Notice that the maximum degree
of $G(k)$ is still at most $\max(\Delta + 1, 3)$, where $\Delta$ is
the maximum degree of $G$.

Given the above definition of $G(k)$, we have the relations
\begin{align}
  Z_S(G(k), \alpha_1, \alpha_2, \lambda) &= (\alpha_2p_k^- +
  p_k^+)^nZ_S(G, \alpha_1,
  \alpha_2, \lambda_k)\label{eq:39}\\
  M_S(G(K), \alpha_1, \alpha_2, \lambda) &= nm_k^- + (m_k^+ -
  m_k^-)M_S(G(k), \alpha_1, \alpha_2, \lambda_k),\label{eq:40}
\end{align}
where $\lambda_k = \lambda(\alpha_1r_k + 1)/(\alpha_2 + r_k)$.  Since
the $r_k$ form a strictly monotone sequence, it follows that (since
$\alpha_1\alpha_2 > 1)$ so do the $\lambda_k$.  In particular, all 
the $\lambda_k$ are distinct. 

\begin{proof}[Proof of Theorem~\ref{thm:hardness-general-two}]
  Proceeding as in the proof of
  Theorem~\ref{thm:ising-partition-hard}, we fix any $\lambda > 0$ and
  $\alpha_1, \alpha_2 > 0$ satisfying $\alpha_1\alpha_1 > 1$, and
  suppose that there exists a polynomial time algorithm $\mathcal{B}$
  which, given a connected graph $H$ of maximum degree at most $\Delta
  \geq 4$, outputs $M_S(G,\alpha_1,\alpha_2, \lambda)$.

  Now consider any connected regular graph $G=(V,E)$ of degree $d :=
  \Delta - 1 \geq 3$ on $n$ vertices.  From the translation in
  eq. (\ref{eq:41}), we see that for any $\lambda > 0$,
   \begin{align*}
     Z_S(G,\alpha_1, \alpha_2,\lambda) &=
     \alpha_2^{|E|}Z_I\inp{G,\beta,\lambda\inp{\frac{\alpha_1}{\alpha_2}}^{d/2}}\text{, and}\\
     M_S(G, \alpha_1,\alpha_2,\lambda) &= M_I\inp{G, \beta, \lambda\inp{\frac{\alpha_1}{\alpha_2}}^{d/2}},
   \end{align*}
   where $\beta = 1/\sqrt{\alpha_1\alpha_2} < 1$.
   Theorem~\ref{thm:rational-interp} along with our main
   Theorem~\ref{thm:lee-yang-ext} then implies that if we can efficiently
   evaluate $M_S(G, \alpha_1, \alpha_2, z)$ at $2n + 2$ distinct
   values of $z$ using our hypothetical algorithm ${\cal B}$, we can
   uniquely determine the coefficients of $Z_S(G, \alpha_1, \alpha_2,
   z)$, and hence also the value of $Z_S(G,\alpha_1, \alpha_2, 1)$, in
   polynomial time.  In view of Theorem~\ref{thm:hardness-regular},
   this would imply that the problem of computing the mean
   magnetization in graphs of maximum degree at most $\Delta$ for
   parameter values $\alpha_1, \alpha_2$ and $\lambda$ is \#P-hard.

   In order to evaluate $M_S(G, \alpha_1, \alpha_2, z)$ at $2n + 2$
   distinct values of $z$, we instead compute
   $M_S(G(k),\alpha_1,\alpha_2, \lambda)$, for $1 \leq k \leq 2n + 2$,
   using our hypothetical algorithm ${\cal B}$.  Notice that this can
   be done since the construction of the $G(k)$ implies that they
   have maximum degrees which are at most one larger than the maximum
   degree of $G$.  Using eqs.\nobreakspace \textup{(\ref {eq:34})}
   to\nobreakspace\textup{~(\ref{eq:37})} and\nobreakspace\textup
   {~(\ref{eq:40})}, and the fact that $m_k^+ - m_k^- > 0$ for all
   $k$, we can then determine $M_S(G,\alpha_1,\alpha_2,\lambda_k)$ in
   polynomial time.  Since $\lambda_k$ is a strictly monotone sequence
   as shown in the discussion above, these evaluations are at distinct
   values of $z$, and hence the reduction is complete.
\end{proof}

\subsection*{Extension to planar graphs}
Cai and Kowalczyk~\cite{cai_dichotomy_2010} also proved the following
planar graph version of Theorem~\ref{thm:hardness-regular}.
\begin{theorem}[{\cite[Theorem 1]{cai_dichotomy_2010}}]
  \label{thm:cai-kowalczyk-planar}
  Fix $\alpha_1, \alpha_2 > 0$ with $\alpha_1\alpha_2 > 1$, $\alpha
  \neq \alpha_2$ and $\Delta \geq 3$.  The problem of computing the
  partition function $Z_S(G, \alpha_1, \alpha_2, 1)$ on
  planar $\Delta$-regular graphs is \#P-hard.
\end{theorem}

In order to extend Theorem~\ref{thm:hardness-general-two} to planar
$\Delta$-regular graphs, we consider the cases $\alpha_1 \neq
\alpha_2$ and $\alpha_1 = \alpha_2 = \alpha$ separately.  In case
$\alpha_1 \neq \alpha_2$, we proceed exactly as in the proof of
Theorem~\ref{thm:hardness-general-two} given above, except that we
start with a \emph{planar} $d$-regular graph $G$ in the reduction, and use
Theorem~\ref{thm:cai-kowalczyk-planar} instead of
Theorem~\ref{thm:hardness-regular} as our starting hardness result.
Since $G$ is planar, so are the $G(k)$, and hence we see that
computing $M_S(H, \alpha_1,\alpha_2,\lambda)$ on planar graphs $H$,
for $\alpha_1, \alpha_2$ and $\lambda$ satisfying the condition
$\alpha_1 \neq\alpha_2$ in addition to the conditions of
Theorem~\ref{thm:hardness-general-two} is \#P-hard. 

We now turn to the case $\alpha_1 = \alpha_2 = \alpha > 1$ (with
$\lambda \neq 1$).  In this case, we start with the fact that computing
$Z_S(G, 2\alpha, \frac{\alpha}{2}, 1)$ on planar $\Delta$-regular
graphs is \#-P hard (this is a direct corollary of
Theorem~\ref{thm:cai-kowalczyk-planar}) .  We again proceed exactly as
in proof of Theorem~\ref{thm:hardness-general-two}, starting with an
arbitrary \emph{planar} $d$-regular graph $G$, and noting that the
$G(k)$ are planar too.  Notice that the proof then shows that assuming
the existence of a polynomial time algorithm to compute the
magnetization in planar graphs of degree at most $d+1$, we can
evaluate the coefficients of the polynomial $Z_s(G, \alpha, \alpha,
z)$, and hence also the quantity $Z_S(G, \alpha, \alpha, 2^d)$.
However we then use the translation to the Ising model given above to
see that
\begin{displaymath}
  Z_S(G,\alpha, \alpha, 2^d) =
  \inp{\frac{\alpha}{2}}^{|E|}Z_S\inp{G, 2\alpha, \frac{\alpha}{2}, 1},
\end{displaymath}
which shows that we can also evaluate $Z_S(G, 2\alpha,
\frac{\alpha}{2}, 1)$ in polynomial time.  This establishes the
\#P-hardness in the remaining case $\alpha_1 = \alpha_2$ (with
$\lambda \neq 1$).  

We thus see that in Theorem~\ref{thm:hardness-general-two}, the input
graphs can be restricted to be planar, and the same hardness result
still holds.

\section{XOR-gadget}
\label{sec:xor-gadget}
\begin{claim}
  The total weight of cycle covers of the XOR-gadget in
  Figure\nobreakspace \ref{fig:xor-gadget} is $2$ when either
  \begin{itemize}
  \item $a$ is connected to an external incoming edge and $d$ is
    connected to an external outgoing edge;
  \end{itemize}
  or
  \begin{itemize}
  \item $a$ is connected to an external outgoing edge and $d$ is
    connected to an external incoming edge.
  \end{itemize}
  For all other external connections of $a$ and $d$, the total weight
  of cycle covers of the gadget is $0$.  
\end{claim}

\begin{proof}
  When the total number of incoming external edges at $a$ and $d$ is
  not equal to the total number of outgoing external edges, the
  XOR-gadget cannot admit a cycle cover due to parity constraints, and
  thus, the total weight of all cycle covers in these cases is
  trivially zero.  A simple way to see this is that a cycle cover
  corresponds to a perfect matching in the natural undirected
  bipartite representation of the gadget discussed above.  When the
  numbers of external incoming and outgoing edges are not equal, the
  bipartite graph does not remain balanced and hence cannot have a
  perfect matching.  For all other configurations, in which the number
  of external incoming and outgoing edges are equal, the weights of
  all cycle covers can be shown to have the claimed value by
  exhaustive enumeration.
\end{proof}

\end{document}